\providecommand{\BibLatexMode}[1]{}
\providecommand{\BibTexMode}[1]{#1}

\ifx\UseBibLatex\undefined%
  \renewcommand{\BibLatexMode}[1]{}
  \renewcommand{\BibTexMode}[1]{#1}
\else
  \renewcommand{\BibLatexMode}[1]{#1}
  \renewcommand{\BibTexMode}[1]{}
\fi

   \documentclass[11pt]{article}%
  \usepackage[cm]{fullpage}
\usepackage{calc}%
\usepackage{ifluatex}
\usepackage{ifxetex}

\ifluatex % LuaTeX
   \usepackage{fontspec}
   \usepackage[utf8]{luainputenc}
\else
   \ifxetex % XeTeX
       \usepackage{fontspec}
   \else
       \usepackage[T1]{fontenc}
       \usepackage[utf8]{inputenc}
   \fi
\fi
\usepackage{mleftright}%
\usepackage{color}%
\usepackage{graphicx}%
\usepackage{wrapfig}%
\usepackage{subfig}%
\usepackage{algorithm}%
\usepackage{algorithmic}%
\usepackage{amsmath}%
\usepackage{bm}%
\usepackage{amssymb}%
\usepackage{amsfonts}%
\usepackage{xspace}%
\usepackage{paralist}%
\usepackage{hyperref}%
\hypersetup{%
   breaklinks,%
   ocgcolorlinks,
   colorlinks=true,%
   urlcolor=[rgb]{0.25,0.0,0.0},
   linkcolor=[rgb]{0.5,0.0,0.0},%
   citecolor=[rgb]{0,0.2,0.445},
   filecolor=[rgb]{0,0,0.4},
   anchorcolor=[rgb]={0.0,0.1,0.2}
}

\usepackage{scalerel}%

   \usepackage[amsmath,thmmarks]{ntheorem}%
   \theoremseparator{.}%

\usepackage{pifont}
   \usepackage{titlesec}%
   \titlelabel{\thetitle. }%
 \numberwithin{figure}{section}%
   \numberwithin{table}{section}%
   \numberwithin{equation}{section}%

\newlength{\savedparindent}
\newcommand{\SaveIndent}{\setlength{\savedparindent}{\parindent}}
\newcommand{\RestoreIndent}{\setlength{\parindent}{\savedparindent}}

   \newtheorem{theorem}{Theorem}[section] % section
   \newtheorem{lemma}[theorem]{Lemma}

   \newtheorem{corollary}[theorem]{Corollary}
   \newtheorem{observation}[theorem]{Observation}

 \theoremheaderfont{\sf}
   \theorembodyfont{\upshape}%
   \theoremstyle{remark}%   
   
\newtheorem{definition}[theorem]{Definition}
   \newtheorem{assumption}[theorem]{Assumption}%  
   \newtheorem*{remark:unnumbered}[theorem]{Remark}%
   \newtheorem{remark}[theorem]{Remark}%

   \theoremheaderfont{\em}%
   \theorembodyfont{\upshape}%
   \theoremstyle{nonumberplain}
   \theoremseparator{}
   \theoremsymbol{\myqedsymbol}
   \newtheorem{proof}{Proof:}

\def\compactify{\itemsep=0pt \topsep=0pt \partopsep=0pt \parsep=0pt}
\let\latexusecounter=\usecounter

{\begin{itemize}%
\setlength{\itemsep}{0pt}%
\setlength{\topsep}{0pt}%
\setlength{\partopsep}{0pt}%
\setlength{\parskip}{0pt}}%
{\end{itemize}}

\renewcommand{\th}{th\xspace}

\DefineNamedColor{named}{OliveGreen} {cmyk}{0.64,0,0.95,0.40}
\providecommand{\ComplexityClass}[1]{{{\textcolor[named]{OliveGreen}{%
            \textsc{#1}}}}}
\providecommand{\NPHard}{{\ComplexityClass{NP-Hard}}%
   \index{NP!hard}\xspace}
\newcommand{\emphic}[2]{%
   \textcolor{blue25}{%
      \textbf{\emph{#1}}}%
   \index{#2}}

\newcommand{\emphi}[1]{\emphic{#1}{#1}}

\newcommand{\myqedsymbol}{\rule{2mm}{2mm}}

\newcommand{\SarielThanks}[1]{%
   \thanks{%
      Department of Computer Science; %
      University of Illinois; %
      201 N. Goodwin Avenue; %
      Urbana, IL, 61801, USA; %
      {\tt \url{http://sarielhp.org}.} %
      #1%
   }
}
\newcommand{\PiotrThanks}[1]{%
   \thanks{%
      Department of Computer Science, MIT; %
      32 Vassar Street, Cambridge, MA, 02139, USA; %
      \texttt{indyk@mit.edu}.%
      #1
   }%
}

\newcommand{\SepidehThanks}[1]{%
   \thanks{%
      Data Science Institute, Columbia University; %
      500 West 120th St., New York, NY, 10027, USA; %
      \texttt{mahabadi@mit.edu}.%
      #1
   }%
}

\newcommand{\IfPrinterVer}[2]{#2}%
\providecommand{\Mh}[1]{{{#1}}}%

\IfFileExists{.latex_printer_friendly}{\def\GenPrinterVer{1}}{}%

\ifx\GenPrinterVer\undefined
   \IfFileExists{.latex_color}{\def\GenColorMath{1}}{} 
\else
   \renewcommand{\IfPrinterVer}[2]{#1}%
\fi

\ifx\GenColorMath\undefined
\else
   \renewcommand{\Mh}[1]{{\textcolor{red}{#1}}}%
\fi

\IfPrinterVer{%
  \definecolor{blue25}{rgb}{0,0,0}
  \DefineNamedColor{named}{RedViolet} {rgb}{0,0,0}%
  \DefineNamedColor{named}{OliveGreen}{rgb}{0,0,0,0}
}{%
  \definecolor{blue25}{rgb}{0,0,0.7}
  \DefineNamedColor{named}{RedViolet} {cmyk}{0.07,0.90,0,0.34}
}%

\newcommand{\HLinkShort}[2]{\hyperref[#2]{#1\ref*{#2}}}
\newcommand{\HLink}[2]{\hyperref[#2]{#1~\ref*{#2}}}
\newcommand{\HLinkPage}[2]{\hyperref[#2]{#1~\ref*{#2}%
      $_\text{p\pageref{#2}}$}}
\newcommand{\HLinkSuffix}[3]{\hyperref[#2]{#1\ref*{#2}{#3}}}
\newcommand{\HLinkPageSuffix}[3]{\hyperref[#2]{#1\ref*{#2}%
      #3$_\text{p\pageref{#2}}$}}
\newcommand{\HLinkPageOnly}[1]{\hyperref[#1]{Page~\refpage*{#1}%
      $_\text{p\pageref{#1}}$}}

\newcommand{\figlab}[1]{\label{fig:#1}}
\newcommand{\figref}[1]{\HLink{Figure}{fig:#1}}

\providecommand{\deflab}[1]{\label{def:#1}}
\newcommand{\defref}[1]{\HLink{Definition}{def:#1}}

\newcommand{\obslab}[1]{\label{observation:#1}}
\newcommand{\obsref}[1]{\HLink{Observation}{observation:#1}}

\providecommand{\eqlab}[1]{}%
\renewcommand{\eqlab}[1]{\label{equation:#1}}
\newcommand{\Eqref}[1]{\HLinkSuffix{Eq.~(}{equation:#1}{)}}

\newcommand{\tbllab}[1]{\label{table:#1}}
\newcommand{\tblref}[1]{Table~\ref{table:#1}}

\newcommand{\asmlab}[1]{\label{assumption:#1}}
\newcommand{\asmref}[1]{\HLink{Assumption}{assumption:#1}}%

\newcommand{\remlab}[1]{\label{rem:#1}}
\newcommand{\remref}[1]{\HLink{Remark}{rem:#1}}%

\newcommand{\corlab}[1]{\label{cor:#1}}
\newcommand{\corref}[1]{\HLink{Corollary}{cor:#1}}%

\newcommand{\lemlab}[1]{\label{lemma:#1}}
\newcommand{\lemref}[1]{\HLink{Lemma}{lemma:#1}}%

\newcommand{\thmlab}[1]{{\label{theo:#1}}}
\newcommand{\thmref}[1]{\HLink{Theorem}{theo:#1}}

\newcommand{\seclab}[1]{\label{sec:#1}}
\newcommand{\secref}[1]{\HLink{Section}{sec:#1}}

\renewcommand{\Re}{\mathbb{R}}%

\newcommand{\ballY}[2]{\Mh{\bm{b}}\pth{#1,#2}}

\newcommand{\Matrix}{\Mh{M}}%

\newcommand{\pth}[2][\!]{\mleft({#2}\mright)}%

\newcommand{\CHX}[1]{\Mh{\mathrm{ConvexHull}}\pth{#1}}%
\newcommand{\spY}[2]{\Mh{{S}}\pth{#1, #2}}%
\newcommand{\tmQueryChar}{\Mh{T_Q}}%
\newcommand{\tmQueryY}[2]{\tmQueryChar\pth{#1, #2}}%

\newcommand{\norm}[1]{\left\|#1\right\|}%

\newcommand{\Cone}{\Mh{C}}%
\newcommand{\Star}{\Mh{\star}}%

\newcommand{\Set}[2]{\left\{ #1 \;\middle\vert\; #2 \right\}}
\newcommand{\Setw}[2]{\{ #1 \;\vert\; #2 \}}

\newcommand{\simp}{\Mh{\triangle}}%

\newcommand{\starY}[2]{\Mh{\mathrm{star}}\pth{#1, #2}}%
\newcommand{\bouqY}[2]{\Mh{\mathrm{bqt}}\pth{#1, #2}}%
\newcommand{\pbouqY}[2]{\Mh{\mathrm{bqt}^+}\pth{#1, #2}}%

\newcommand{\bouqA}{\Mh{\mathcal{B}}}%

\newcommand{\SimSetC}{\Mh{\bm{\Delta}}}%
\newcommand{\SimSetI}[2]{\SimSetC\pth{ #1, #2}}%

\newcommand{\SimSetD}{\Mh{\bm{\Delta}}_\FlatC}%

\newcommand{\simX}[1]{\Mh{\simp}_{#1}}%
\newcommand{\simCY}[2]{\Mh{\simp}_{\FlatC}\pth{#1,#2}}%
\newcommand{\simCX}[1]{\Mh{\simp}_{\FlatC}\pth{#1}}%

\newcommand{\radBX}[1]{\Mh{\mathrm{r}}^{}_\Pb\pth{#1}}%
\newcommand{\wdX}[1]{\Mh{{d}\pth{#1}}}

\newcommand{\vLen}{\Mh{\bm{\ell}}}%
\newcommand{\pFlatY}[2]{\flat^+\pth{#1, #2}}

\newcommand{\FlatB}{\Mh{F}}%
\newcommand{\FlatC}{\Mh{G}}%
\newcommand{\FlatD}{\Mh{H}}%

\newcommand{\dimX}[1]{\mathrm{dim}\pth{#1}}%
\newcommand{\flatF}{\Mh{F}}%
\newcommand{\flatG}{\Mh{G}}%
\newcommand{\flatH}{\Mh{H}}%

\newcommand{\pFlat}[1]{\pnt^{}_{\FlatC}\pth{#1}}%

\newcommand{\qB}{\query^{}_\Pb}%

\newcommand{\qX}[1]{\query_{\FlatC}\pth{#1}}%
\newcommand{\qAX}[1]{\query_{\FlatD}\pth{#1}}%

\newcommand{\intX}[1]{\mathrm{int}\pth{#1}}

\renewcommand{\flat}{\Mh{\mathsf{f}}}%
\newcommand{\flatA}{\Mh{\mathsf{g}}}%

\newcommand{\flatX}[1]{\flat_{#1}}%

\newcommand{\eps}{\varepsilon}%

\newcommand{\SetX}{\Mh{X}}
\newcommand{\SetY}{\Mh{Y}}

\newcommand{\Term}[1]{\textsf{#1}}%
\newcommand{\SLR}{\Term{SLR}\xspace}%
\newcommand{\ANN}{\Term{ANN}\xspace}%

\newcommand{\ANIS}{\Term{ANIS}\xspace}%
\newcommand{\ANIF}{\Term{ANIF}\xspace}%
\newcommand{\ANLF}{\Term{ANLIF}\xspace}%

\newcommand{\remove}[1]{}%

\newcommand{\distSet}[2]{\Mh{\mathsf{d}}(#1, #2)}%

\newcommand{\distY}[2]{\norm{#1- #2}}%
\newcommand{\distwY}[2]{\|{#1- #2}\|}%

\newcommand{\brc}[1]{\left\{ {#1} \right\}}%
\newcommand{\Angle}[1]{\measuredangle #1}%

\newcommand{\NNV}[1]{#1^{\Mh{*}}}%
\newcommand{\nnq}{\NNV{\query}}%
\newcommand{\qf}{\query^{}_{\flatA}}

\newcommand{\etal}{\textit{et~al.}\xspace}

\newcommand{\vecA}{\Mh{\tau}}
\newcommand{\vecB}{\Mh{u}}

\newcommand{\projY}[2]{\Mh{\mathsf{proj}}^{}_{#1}\pth{#2}}

\newcommand{\nnY}[2]{\mathsf{nn}\pth{#1,#2}}
\newcommand{\dirY}[2]{\mathsf{dir}(#1,#2)}

\newcommand{\PrismB}{\Mh{\Phi}_{\Pb}}%
\newcommand{\PrismBX}[1]{\Mh{\Phi}_{\Pb \cup \brc{\pnt_k}}}%

\newcommand{\rad}{\Mh{r}}

\newcommand{\Pb}{\Mh{B}}%
\newcommand{\affX}[1]{\mathrm{aff}\pth{#1}}

\newcommand{\PntSet}{\Mh{P}}%
\newcommand{\PntSetA}{\Mh{V}}%
\newcommand{\PntSetB}{\Mh{U}}%
\newcommand{\PntSetC}{\Mh{Z}}%

\newcommand{\query}{\Mh{q}}%
\newcommand{\wquery}{\Mh{\widehat{q}}}%

\newcommand{\cen}{\Mh{c}}%

\newcommand{\bpnt}{\Mh{b}}%

\newcommand{\pnt}{\Mh{p}}%

\newcommand{\pntA}{\Mh{y}}%
\newcommand{\pntB}{\Mh{u}}%
\newcommand{\pntC}{\Mh{v}}%
\newcommand{\pntD}{\Mh{\bm{s}}}%
\newcommand{\pntE}{\Mh{x}}%

\newcommand{\normY}[2]{\| #1 - #2 \|}
\newcommand{\SphereChar}{\Mh{\ensuremath{\mathbb{S}}}}
\newcommand{\SphereY}[2]{\SphereChar\pth{#1, #2}}
\newcommand{\SphereA}{\SphereChar}
\newcommand{\LineY}[2]{\mathrm{line}\pth{#1, #2}}%
\newcommand{\refX}[1]{\mathrm{ref}\pth{#1}}%
\newcommand{\projSX}[1]{{ #1}^{}_{|\SphereChar}}

\newcommand{\origin}{\Mh{\textbf{0}}}%

\newcommand{\PRC}{\Mh{\circledcirc}}%\text{\ding{105}}

\newcommand{\PRZ}[3]{\PRC \pth{#1, #2, #3 }}%

\newcommand{\cardin}[1]{\left| {#1} \right|}%

\newcommand{\TrC}{\Mh{M}}%
\newcommand{\TrX}[1]{\TrC\pth{#1}}%
\newcommand{\ts}{\hspace{0.6pt}}
\newcommand{\DS}{\ensuremath{\Mh{\mathcal{D}}}\xspace}%

\newcommand{\FlatSet}{\Mh{\mathcal{F}}}%

\newcommand{\FlatsY}[2]{\Mh{\mathcal{F}_{#1}\pth{#2}}}%
\newcommand{\SimpsY}[2]{\Mh{\bm{\Delta}_{#1}\pth{#2}}}%
\newcommand{\LFlatsY}[2]{\Mh{\mathcal{L}_{#1}\pth{#2}}}%
\newcommand{\Family}{\Mh{\mathcal{G}}}%
\renewcommand{\th}{th\xspace}

\newcommand{\ba}{\Mh{\alpha}}%
\newcommand{\bAngleY}[2]{\Mh{\bm{\alpha}}^{}_{#1}\pth{#2}}

\newcommand{\len}{\Mh{\ell}}%

\newcommand{\ocX}[1]{{#1}^{\bot}}%
\newcommand{\DotProd}[2]{\permut{{#1},{#2}}}%
\newcommand{\permut}[1]{\left\langle {#1} \right\rangle}%

\newcommand{\PASet}{\PntSet_{\Mh{\!\measuredangle}}}

\def\mypar#1{\medskip\noindent\textbf{#1}}

\newcommand{\tldO}{\scalerel*{\widetilde{O}}{j^2}}%

\BibLatexMode{%
   \usepackage[bibencoding=ascii,style=alphabetic,backend=biber]{biblatex}%
   \usepackage{sariel_biblatex}%   
   
}

\title{Approximate Sparse Linear Regression%
   \footnote{The paper is about to appear in ICALP'18. This research was supported by NSF and Simons Foundation.}%
}

   \author{%
      Sariel Har-Peled%
      \SarielThanks{%
         Work on this paper was partially supported by NSF AF awards %
         CCF-1421231, and % Started June 2014
         CCF-1217462.  % Started June 2012
      }%
      \and%
      Piotr Indyk\PiotrThanks{}%
      \and%
      Sepideh Mahabadi\SepidehThanks{This work was done while this author was at MIT.}%
   }%

\begin{document}

\maketitle

\begin{abstract}
    In the \emph{Sparse Linear Regression} (\SLR) problem, given a
    $d \times n$ matrix $\Matrix$ and a $d$-dimensional query
    $\query$, the goal is to compute a $k$-sparse $n$-dimensional vector
    $\vecA$ such that the error $\norm{\Matrix \vecA-\query}$ is
    minimized. This problem is equivalent to the following geometric problem:
    given a set $P$ of $n$ points and a query point $\query$ in $d$ dimensions, 
    find the closest $k$-dimensional subspace to $\query$,
    that is spanned by a subset of $k$ points in $P$.
    In this paper, we present data-structures/algorithms
    and conditional lower bounds for several variants  of this
    problem (such as finding the closest induced $k$ dimensional flat/simplex instead of a subspace).

    In particular, we present \emph{approximation} algorithms for the
    online variants of the above problems with query time
    $\tldO(n^{k-1})$, which are of interest in the "low sparsity regime" where $k$ is
    small, e.g., $2$ or $3$. For $k=d$, this matches, up to
    polylogarithmic factors, the lower bound that relies on the
    \emph{affinely degenerate conjecture} (i.e., deciding if $n$
    points in $\Re^d$ contains $d+1$ points contained in a hyperplane takes
    $\Omega(n^d)$ time).
    % We show also somewhat weaker lower bounds for general 
    %values of $k$
    %that rely on the Hopcroft's problem, and the $k$-sum problem.
 Moreover, our algorithms involve formulating and
solving several geometric subproblems, which we believe to
be of independent interest.    
\end{abstract}

%\newpage

\setcounter{page}{1}
\pagenumbering{arabic}

\section{Introduction}
The goal of the {\em Sparse Linear Regression} (\SLR) problem is to
find a sparse linear model explaining a given set of
observations. Formally, we are given a matrix
$\Matrix \in \Re^{d\times n}$, and a vector $\query \in \Re^d$, and
the goal is to find a vector $\vecA$ that is $k$-sparse (has at most
$k$ non-zero entries) and that minimizes
$\norm{\query - \Matrix \vecA}_2$. The problem also has a natural
query/online variant where the matrix $\Matrix$ is given in advance
(so that it can be preprocessed) and the goal is to quickly find
$\vecA$ given $\query$.

Various variants of \SLR have been extensively studied, in a wide range
of fields including
\begin{inparaenum}[(i)]
    \item statistics and machine learning \cite{t-rssvl-96,
       t-rssvl-11},
    \item compressed sensing~\cite{d-cs-06}, and
    \item computer vision~\cite{wygsm-rfrsr-09}.
\end{inparaenum}
The query/online variant is of particular interest in the application
described by Wright \etal \cite{wygsm-rfrsr-09}, where the matrix
$\Matrix$ describes a set of image examples with known labels and
$\query$ is a new image that the algorithm wants to label.

If the matrix $\Matrix$ is generated at random or satisfies certain
assumptions, it is known that a natural convex relaxation of the
problem finds the optimum solution in polynomial
time~\cite{crt-rupes-06, cds-adbp-98}. However, in general the problem
is known to be \NPHard \cite{n-sasls-95, dma-aga-97}, and even hard to
approximate up to a polynomial factor \cite{fkt-vsh-15} (see below for
a more detailed discussion).  Thus, it is likely that any algorithm
for this problem that guarantees "low" approximation factor must run
in exponential time. A simple upper bound for the offline problem is
obtained by enumerating ${n \choose k}$ possible supports of $\vecA$
and then solving an instance of the $d \times k$ least squares problem.
This results in $n^k (d+k)^{O(1)}$ running time, which (to the best of
our knowledge) constitutes the fastest known algorithm for this
problem.  At the same time, one can test whether a given set of $n$
points in a $d$-dimensional space is degenerate by reducing it to $n$
instances of \SLR with sparsity $d$. The former problem is conjectured
to require $\Omega(n^d)$ time~\cite{es-blbda-95} -- this is the
\emph{affinely degenerate conjecture}. This provides a natural barrier
for running time improvements (we elaborate on this below in
\secref{l:b}).

In this paper, we study the complexity of the problem in the case
where the sparsity parameter $k$ is constant. In addition to the
formulation above, we also consider two more constrained variants of
the problem. First, we consider the \emphi{Affine SLR} where the
vector $\vecA$ is required to satisfy $\norm{\vecA}_1 =1$, and second,
we consider the \emphi{Convex SLR} where additionally $\vecA$ should
be non-negative.  We focus on the approximate version of these
problems, where the algorithm is allowed to output a $k$-sparse vector
$\vecA'$ such that $\norm{\Matrix \vecA'-\query}_2$ is within a factor
of $1+\eps$ of the optimum.

\mypar{Geometric interpretation.} The \SLR problem is equivalent to the
\emph{Nearest Linear Induced Flat} problem defined as follows. Given a
set $\PntSet$ of $n$ points in $d$ dimensions and a $d$-dimensional
vector $\query$, the task is to find a $k$-dimensional flat spanning a
subset $B$ of $k$ points in $\PntSet$ and the origin, such that the
(Euclidean) distance from $\query$ to the flat is minimized.  The
Affine and Convex variants of \SLR respectively correspond to finding
the \emph{Nearest Induced Flat} and the \emph{Nearest Induced Simplex}
problems, where the goal is to find the closest $(k-1)$-dimensional
flat/simplex spanned by a subset of $k$ points in $\PntSet$ to the
query.

%%%%%%%%%%%%%%%%%%%%%%%%%%%%%%%%%%%%%%%%%%%%%%%%%%%%%%%%%%%%%%%%%

%SARIEL: start new
\mypar{Motivation for the problems studied.}
Given a large\footnote{{\it Bigger than the biggest thing ever and
      then some.  Much bigger than that in fact, really amazingly
      immense, a totally stunning size, ''wow, that's big'', time.} --
   The Restaurant at the End of the Universe, Douglas Adams.}  set of
items (e.g., images), one would like to store them efficiently for
various purposes. One option is to pick a relatively smaller subset of
representative items (i.e., support vectors), and represent all items
as a combination of this \emph{supporting set}. Note, that if our
data-set is diverse and is made out of several distinct groups (say,
images of the sky, and images of children), then naturally, the data
items would use only some of the supporting set for representation
(i.e., the representation over the supporting set is naturally
sparse). As such, it is natural to ask for a sparse representation of
each item over the (sparse but still relatively large) supporting set.
(As a side note, surprisingly little is known about how to choose such
a supporting set in theory, and the problem seems to be surprisingly
hard even for points in the plane.)

Now, when a new item arrives to the system, the task is to compute its
best sparse representation using the supporting set, and we would like
to do this as fast as possible (which admittedly is not going to be
that fast, see below for details).

% SARIEL: end new

%%%%%%%%%%%%%%%%%%%%%%%%%%%%%%%%%%%%%%%%%%%%%%%%%%%%%%%%%%%%%%%%%%%%%%%
\begin{table}
    \centerline{%
       \begin{tabular}{|c|l|c|c|c|}
         \cline{2-5}
         \multicolumn{1}{c||}{}
         & Comment
         & Space
         & Query
         & See
         \\
         \hline\hline
         %%%%%%%%%%%%%%%%%%%% 
         \SLR
         &
         &
           $\Bigl. n^{k-1}\spY{n}{\eps}$
         &
           $n^{k-1} \tmQueryY{n}{\eps}$
         &
           %$\secref{sec:nearest-flat}
           \thmref{slr-k2}%
         \\
         \hline
         %%%%%%%%%%%%%%%%%%%% 
         Affine \SLR%
         &%
         &%
           $\Bigl. n^{k-1}\spY{n}{\eps}$
         &%
           $n^{k-1} \tmQueryY{n}{\eps}$
         &%
%           \secref{sec:nearest-flat}
           \thmref{slr-k}%
         \\
         \hline
         %%%%%%%%%%%%%%%%%%%% 
         Convex \SLR
         &%
         &%
           $\Bigl. n^{k-1}\spY{n}{\eps}\log ^{k} n$
         &%
           $n^{k-1} \tmQueryY{n}{\eps} \log ^{k} n$%
         &%
         % \secref{sec:nearest-simplex}
           \lemref{slr-k:simplex}%
         \\
         \cline{2-5}
         %%%%%%%%%%%%%%%%%%%% 
         &%
           $k=2$ \& $\eps\leq 1$ 
         &%
           $\Bigl.n\spY{n}{\eps}\log n$
         &%
           $n\tmQueryY{n}{\eps}\eps^{-2} \log n$
         &%
         % \secref{sec:seg}
           \thmref{a:n:n:i:segment}%
         \\
         \cline{1-5}
         %%%%%%%%%%%%%%%%%%%% 
         \begin{minipage}[c]{2.3cm}
             \begin{center}
                 Approximate nearest
             \end{center}
         \end{minipage}
         &%
           \begin{minipage}{3cm}
               \smallskip%
               $k=2$ \\
               $\bigl.2(1+\eps)$ Approx
           \end{minipage}
         &%
           \multicolumn{2}{|c|}{%
           \begin{minipage}{3cm}
               \centerline{$n^{1+O(\frac{1}{(1+\eps)^2})}$}
           \end{minipage}
           }
         &
         % \secref{offline}
           \thmref{n:i:segment}
           %%%%%%%%%%%%%%%%%%%% 
         \\%
         %%%%%%%%%%%%%%%%%%%% 
         \cline{2-5}%
         { induced segment}
         &
           %\begin{minipage}{3cm}
           $d=O(1)$
               % \end{minipage}
         &           
           \multicolumn{2}{|c|}{%
           $\Bigl.O( n \log n + n/\eps^d)$
           }
         &
           \corref{n:i:s:low:dim}%
         \\
         \hline
       \end{tabular}%
\vspace{0.1cm}
        }%
    \caption[summary]{Summary of results. Here,
       $\spY{n}{\eps}$ denotes the preprocessing time and space used
       by a
       $(1+\eps)$-\ANN (approximate nearest-neighbor) data-structure,
       and
       $\tmQueryY{n}{\eps}$ denotes the query time (we assume all
       these bounds are at least linear in the dimension
       $d$).  All the data-structures, except the last one, provide
       $(1+\eps)$-approximation.
In the nearest induced segment case (i.e., this is the
       offline convex \SLR case) the algorithm answers a single query
       -- in the specific case of \thmref{n:i:segment}, the query time
       is slighly better than the online variant, but significnatly,
       the space is better by a factor of
       $n$.  Unfortunately, the quality of apprxoimation is worse.  \vspace{-0.5cm}}%
    \tbllab{results}%
\end{table}
%%%%%%%%%%%%%%%%%%%%%%%%%%%%%%%%%%%%%%%%%%%%%%%%%%%%%%%%%%%%%%%%%%%%%%%

% ......................................... Results
\subsection{Our results}
% {\color{red}
%In this paper, we present upper and lower bounds\footnote{%
   % \color{red}
%   Some of our results involve randomized algorithms. See the
%   statements of theorems for the details.}  for the aforementioned
%problems.

\mypar{Data-structures.} %
    We present data-structures to solve the online variants of the
    \SLR, Affine \SLR and Convex \SLR problems, for general value of
    $k$.  Our algorithms use a provided approximate nearest-neighbor
    (\ANN) data-structure as a black box.  The new results are
    summarized in \tblref{results}.
    For example, for a point set in constant dimension, we get a near
    linear time algorithm that computes for a single point, the
    nearest induced segment to it in near linear time, see
    \corref{n:i:s:low:dim}.

For small values of $k$, our algorithms offer notable improvements of the query time over the aforementioned naive algorithm, albeit at a cost of preprocessing. Below in \secref{l:b}, we show how our result matches the lower bound that relies on the affinely degenerate conjecture. Moreover, our algorithms involve formulating and
solving several interesting geometric subproblems, which we believe to
be of independent interest.

    \mypar{Conditional lower bound.}
    We show a
    conditional lower bound of
    $\Omega(n^{k/2}/(e^k \log^{\Theta(1)} n))$, for the offline
    variants of all three problems. Improving this lower bound further, for the
    case of $k=4$, would imply a nontrivial lower bound for famous
    Hopcroft's problem. See \secref{sec:hopcroft} for the description.
    Our lower bound result presented in \secref{lowerbound}, follows
    by a reduction from the $k$-sum problem which is conjectured to
    require $\Omega(n^{\lceil k/2 \rceil}/\log^{\Theta(1)} n)$ time
    (see e.g., \cite{pw-pfsa-10}, Section 5). This provides further
    evidence that the off-line variants of the problem require $n^{\Omega(k)}$ time.

\subsubsection{Detecting affine degeneracy}
\seclab{l:b}
Given a point set $\PntSet$ in $\Re^d$ (here $d$ is conceptually
small), it is natural to ask if the points are in general position --
that is, all subsets of $d+1$ points are affinely independent.  The
\emph{affinely degenerate conjecture} states that this problem
requires $\Omega(n^{d})$ time to solve \cite{es-blbda-95}. This can be
achieved by building the arrangement of hyperplanes in the dual, and
detecting any vertex that has $d+1$ hyperplanes passing through
it. This problem is also solvable using our data-structure. (We note that since the approximation of our data structure is multiplicative, and in the reduction we only need to detect distance of $0$ from larger than $0$, we are able to solve the exact degeneracy problem as described next). Indeed,
we instantiate \thmref{slr-k}, for $k=d$, and using a low-dimensional
$(1+\eps)$-\ANN data-structure of Arya \etal
\cite{amnsw-oaann-98}. Such an \ANN data-structure uses
$\spY{n}{\eps} = O(n)$ space, $O(n \log n)$ preprocessing time, and
$\tmQueryY{n}{\eps}=O( \log n + 1/\eps^d) =O(\log n)$ query time (for a fixed
constant $\eps<1$). Thus, by \thmref{slr-k}, our data structure has total space usage and preprocessing time of $\tilde O (n^k)$ and a query time of $\tilde O(n^{k-1})$. Detecting affine degeneracy then reduces to solving
for each point of $\query \in \PntSet$, the problem of finding the
closest $(d-1)$-dimensional induced flat (i.e., passing through $d$ points) of
$\PntSet \setminus \brc{ \query}$ to $\query$. It is easy to show that
this can be solved using our data-structure with an extra $\log$
factor\footnote{The details are somewhat tedious -- one generates
   $O( \log n)$ random samples of $\PntSet$ where each point is picked
   with probability half. Now, we build the data-structure for each of
   the random samples. With high probability, for each of the query
   point $\query \in \PntSet$, one of the samples contains, with high
   probability, the $d$ points defining the closest flat, while not
   containing $\query$.}.    This means that the total runtime (including the preprocessing and the $n$ queries) will be $\tilde O (n^k)= \tilde O(n^d)$. Thus, up to polylogarithmic factor, the
data-structure of \thmref{slr-k} provides an optimal trade-off under the affinely
degenerate conjecture. We emphasize that this reduction only rules out the existence of algorithms  for online variants of our problems that improve both the preprocessing time from $O(n^{k})$, and query time from $O(n^{k-1})$ by much; it does not rule out for example the algorithms with large preprocessing time (in fact much larger than $n^{k}$) but small query time. 

% .............................. Related Work
\vspace{-0.1cm}
\subsection{Related work}

The computational complexity of the approximate sparse linear
regression problem has been studied, e.g., in~\cite{n-sasls-95,
   dma-aga-97, fkt-vsh-15}. In particular, the last paper proved a
strong hardness result, showing that the problem is hard even if the
algorithm is allowed to output a solution with sparsity
$k'=k 2^{\log^{1-\delta} n}$ whose error is within a factor of
$n^c m^{1-\alpha}$ from the optimum, for any constants
$\delta, \alpha>0$ and $c>1$.

The query/online version of the Affine \SLR problem can be reduced to
the {\em Nearest $k$-flat Search Problem} studied in \cite{magen-drpvd-02, bhz-anss-11, m-anlsh-15}, where the database consists
of a set of $k$-flats (affine subspaces) of size $N$ and the goal is
to find the closest $k$-flat to a given query point $\query$. Let
$\PntSet$ be a set of $n$ points in $\Re^d$ that correspond to the
columns of $\Matrix$. The reduction proceeds by creating a database of
all $N=\binom{n}{k}$ possible $k$-flats that pass through $k$ points
of $\PntSet$. 
However, the result of \cite{bhz-anss-11} does not
provide multiplicative approximation guarantees, although it does provide some alternative guarantees and has  been  validated  by  several  experiments. The  result of \cite{magen-drpvd-02}, provides provable guarantees and fast query time of
$(d+\log N + 1/\eps)^{O(1)}$, but the space
requirement is quasi-polynomial of the form $2^{(\log N)^{O(1)}} = 2^{(k\log n)^{O(1)}}$.
Finally the result of \cite{m-anlsh-15} only works for the special case of $k=2$, 
%it is known that there exists an
%algorithm that uses $(dn/\eps)^{O(1)}$ space and
%$(1/\eps+d+\log n)^{O(1)}$ query time. This 
and yields an algorithm with
space usage
\begin{math}
    O\pth{ n^{14} \eps^{-3} \spY{n^2}{\eps}}
\end{math}
and query time $O\pth{ \tmQueryY{{n^2\eps^{-4}}}{\eps} \log^{2} n }$%
\footnote{%
   The exact exponent is not specified in the main theorem of
   \cite{m-anlsh-15} and it was obtained by an inspection of the
   proofs in that paper.}.  Similar results can be achieved for the
other variants.% of \SLR.

The \SLR problem has a close relationship with the \textit{Approximate
   Nearest Neighbor (\ANN)} problem. In this problem, we are given a
collection of $N$ points, and the goal is to build a data structure
which, given any query point $\query$, reports the data point whose
distance to the query is within a $(1+\eps)$ factor of the distance of
the closest point to the query.  There are many efficient algorithms
known for the latter problem.  One of the state of the art results for
\ANN in Euclidean space answers queries in time
$(d \log (N)/\eps^2)^{O(1)}$ using $(dN)^{O(1/\eps^2)}$
space~\cite{kor-esann-00,im-anntr-98}.

% .........................................techniques
\subsection{Our techniques and sketch of the algorithms}

\mypar{Affine \SLR (nearest flat).} %
To solve this problem, we first fix a subset $\Pb\subseteq P$ of $k-1$
points, and search for the closest $(k-1)$-flat among those that
contain $\Pb$. Note, that there are at most $n-k+1$ such flats. Each
such flat $\flat$, as well as the query flat $Q_{\mathrm{flat}}$
(containing $\Pb$ and the query $\query$), has only one additional
degree of freedom, which is represented by a vector $v_H$ ($v_Q$,
resp.) in a $d-k+1$ space. The vector $v_H$ that is closest to $v_Q$
corresponds to the flat that is closest to $\query$. This can be found
approximately using standard \ANN data structure, resulting in an
algorithm with running time $O(n^{k-1}\cdot \tmQueryY{n}{\eps})$.
Similarly, by adding the origin to the set $\Pb$, we could solve the
\SLR problem in a similar way.

% Finding the closest $(k-1)$-flat (i.e., the closest
% $(k-1)$-dimensional affine subspace) to the query point can be
% easily achieved in $O(n^{k-1}\cdot \tmQueryY{n}{\eps})$ time as
% follows.

\mypar{Convex \SLR (nearest simplex).} %
This case requires an intricate combination of low and high
dimensional data structures, and is the most challenging part of
this work.  To find the closest $(k-1)$-dimensional induced {\em
   simplex}, one approach would be to fix $\Pb$ as before, and find
the closest corresponding flat.  This will work only if the
projection of the query onto the closest flat falls inside of
its corresponding simplex. Because of that, we need to restrict our
search to the flats of {\em feasible simplices}, i.e., the simplices
$S$ such that the projection of the query point onto the corresponding
flat falls inside $S$.  If we manage to find this set, we can use the
algorithm for affine \SLR to find the closest one. Note that finding
the distance of the query to the closest non-feasible simplex can
easily be computed in time $n^{k-1}$ as the closest point of such a simplex to the query
lies on its boundary which is a lower dimensional object.

% {\color{red}
Let $S$ be the unique simplex obtained from $\Pb$ and an additional
point $\pnt$. Then we can determine whether $S$ is feasible or not
only by looking at (i) the relative positioning of $\pnt$ with respect
to $\Pb$, that is, how the simplex looks like in the flat going
through $S$, (ii) the relative positioning of $\query$ with respect to
$\Pb$, and (iii) the distance between the query and the flat of the
simplex.  Thus, if we were given a set of simplices through $\Pb$ such that all
their flats were at a distance $r$ from the query, we could build a
single data structure for retrieving all the feasible flats. This
can be done by mapping all of them in advance onto a unified $(k-1)$
dimensional space (the ``parameterized space''), and then using $k-1$
dimensional orthogonal range-searching trees in that space.

However, the minimum distance $r$ is not known in general.  Fortunately, as we
show the feasibility property is monotone in the distance: the
farther the flat of the simplex is from the query point, the weaker
constraints it needs to satisfy. Thus, given a threshold value $r$,
our algorithm retrieves the simplices satisfying the restrictions they
need to satisfy if they were at a distance $r$ from the query. This
allows us to use binary search for finding the right value of $r$ by
random sampling.  The final challenge is that, since our access is to
an approximate NN data structure (and not an exact one), the above
procedure yields a superset of feasible simplices. The algorithm then
finds the closest flat corresponding to the simplices in this
superset. We show that although the reported simplex may not be
feasible (the projection of $\query$ on to the corresponding flat does
not fall inside the simplex), its distance to the query is still
approximately at most $r$.

% \textbf{Nearest Segment.}
 
\mypar{Convex \SLR for $k=2$ (offline nearest segment).} %
Given a set of $n$ points $\PntSet$ and the query $\query$ all at
once, the goal is to find the segment composed of two points in
$\PntSet$ that is closest to the query, in sub-quadratic time.  To
this end, we project the point set to a sphere around the query point,
and preprocess it for \ANN. For each point of $\PntSet$, we use its
reflected point on the sphere, to perform an \ANN query, and get another
(projected) point. Taking the original two points, generates a
segment. We prove that among these $n$ segments, there is a
$O(1)$-\ANN to the closest induced segment.  This requires building
\ANN data-structure and answering $n$ \ANN queries, thus yielding a
subquadratic-time algorithm.  See \secref{offline} for details.

\mypar{Conditional Lower bound.} %
Our reduction from $k$-sum is randomized, and based on the following
idea. First observe that by testing whether the solution to \SLR is
zero or non-zero we can test whether there is a subset of $k$ numbers
and a set of associated $k$ weights such that the {\em weighted} sum
is equal to zero. In order to solve $k$-sum, however, we need to
force the weights to be equal to $1$.  To this end, we lift the
numbers to vectors, by extending each number by a vector selected at
random from the standard basis $e_1 \ldots e_k$. We then ensure that
in the selected set of $k$ numbers, each element from the basis is
represented exactly once and with weight 1. This ensures that the
solution to \SLR yields a feasible solution to $k$-sum.

%%%%%%%%%%%%%%%%%%%%%%%%%%%%%%%%%%%%%%%%%%%%%%%%%%%%%%%%%%%%%%%%%%%%%%%%
%%%%%%%%%%%%%%%%%%%%%%%%%%%%%%%%%%%%%%%%%%%%%%%%%%%%%%%%%%%%%%%%%%%%%%%%
%%%%%%%%%%%%%%%%%%%%%%%%%%%%%%%%%%%%%%%%%%%%%%%%%%%%%%%%%%%%%%%%%%%%%%%%
\section{Preliminaries}

\subsection{Notations}
Throughout the paper, we assume $\PntSet \subseteq \Re^d$ is the set
of input points which is of size $n$.  In this paper, for simplicity,
we assume that the point-sets are non-degenerate, however this
assumption is not necessary for the algorithms.  We use the notation
$X \subset_{i} B$ to denote that $X$ is a subset of $B$ of size $i$,
and use $\origin$ to denote the origin.  For two points
$\pntA,\pntB\in \Re^d$, the segment the form is denoted by
$\pntA \pntB$, and the line formed by them by $\LineY{\pntA}{\pntB}$.

\begin{definition}
    \deflab{simplex}%
    For a set of points $S$, let
    \begin{math}
        \flat_S = \affX{S} = \Set{\sum_{i=1}^{\cardin{S}} \alpha_i \pnt_i}{
           \pnt_i \in S, \text{ and } \sum_{i=1}^{\cardin{S}} \alpha_i = 1}
    \end{math}
    be the $(\cardin{S}-1)$-dimensional flat (or
    \emphi{$(\cardin{S}-1)$-flat} for short) passing through the
    points in the set $S$ (aka the \emph{affine hull} of $S$).  The
    $(\cardin{S}-1)$-dimensional simplex (\emphi{$({\cardin{S}}-1)$-simplex} for
    short) that is formed by the convex-hull of the points of $S$ is
    denoted by $\simX{S}$. We denote the \emphi{interior} of a simplex
    $\simX{S}$ by $\intX{\simX{S}}$.
\end{definition}

\begin{definition}[distance and nearest-neighbor]
    For a point $\query \in \Re^d$, and a point $\pnt \in \Re^d$, we
    use $\distSet{ \query}{\pnt} = \distY{\query}{\pnt}_2$ to denote
    the \emphi{distance} between $\query$ and $\pnt$. For a closed set
    $\SetX \subseteq \Re^d$, we denote by
    $\distSet{\query}{\SetX} = \min_{\pnt \in \SetX}
    \distY{\query}{\pnt}_2$ the \emphi{distance} between $\query$ and
    $\SetX$.  The point of $\SetX$ realizing the distance between
    $\query$ and $\SetX$ is the \emphi{nearest neighbor} to $\query$
    in $\SetX$, denoted by $\nnY{\query}{\SetX}$. We sometimes refer
    to $\nnY{\query}{\SetX}$ as the \emph{projection} of $\query$ onto
    $\SetX$.

    More generally, given a finite family of such sets
    $\Family = \Set{\SetX_i \subseteq \Re^d }{i=1,\ldots, m}$, the
    \emphi{distance} of $\query$ from $\Family$ is
    $\distSet{\query}{\Family} = \min_{ \SetX \in \Family}
    \distSet{\query}{\SetX}$. The \emph{nearest-neighbor}
    $\nnY{\query}{\Family}$ is defined analogously to the above.
\end{definition}

\begin{assumption}
    \asmlab{ANN}%
    Throughout the paper, we assume we have access to a data structure
    that can answer $(1+\eps)$-\ANN queries on a set of $n$ points in
    $\Re^d$.  We use $\spY{n}{\eps}$ to denote the space requirement
    of this data structure, and by $\tmQueryY{n}{\eps}$ to denote the
    query time.
\end{assumption}

\subsubsection{Induced stars, bouquets, books, simplices and flats}

\begin{definition}
    \deflab{star}%
    Given a point $\bpnt$ and a set $\PntSet$ of points in $\Re^d$,
    the \emphi{star} of $\PntSet$, with the base $\bpnt$, is the set
    of segments
    \begin{math}
        \starY{\bpnt}{\PntSet}%
        =%
        \Set{ \bpnt\pnt}{ \pnt \in \PntSet \setminus \brc{\bpnt}}.
    \end{math}
    Similarly, given a set $\Pb$ of points in $\Re^d$, with
    $\cardin{\Pb} = k-1 \leq d$, the \emphi{book} of $\PntSet$, with
    the base $\Pb$, is the set of simplices
    \begin{math}
        \SimSetI{\Pb}{\PntSet}%
        =%
        \Set{ \simX{\Pb\cup\brc{\pnt}} }{ \pnt \in \PntSet \setminus
           \Pb}.
    \end{math}
    Finally, the set of flats induced by these simplices, is the
    \emphi{bouquet} of $\PntSet$, denoted by
    \begin{math}
        \bouqY{\Pb}{\PntSet}%
        =%
        \Set{ \flatX{\Pb\cup\brc{\pnt}} }{ \pnt \in \PntSet \setminus
           \Pb}.
    \end{math}
\end{definition}
If $\Pb$ is a single point, then the corresponding book is a star, and the corresponding bouquet is a
set of lines all passing through the single point in $\Pb$.

\begin{definition}
    \deflab{induced:flats}%
    For a set $\PntSet \subseteq \Re^d$, let
    $\LFlatsY{k}{\PntSet} = \Set{ \flat_{S\cup\{\origin\}}}{ S
       \subset_{k} \PntSet}$ be the set of all linear $k$-dimensional
    subspaces induced by ${\PntSet}$, and
    $\FlatsY{k}{\PntSet} = \Set{ \flat_S}{ S \subset_{k} \PntSet}$ be
    the set of all $(k-1)$-flats induced by $\PntSet$. Similarly, let
    $\SimpsY{k}{\PntSet} = \Set{ \simX{S}}{ S \subset_{k} \PntSet}$ be
    the set of all $(k-1)$-simplices induced by $\PntSet$.
\end{definition}

\subsection{Problems}

In the following, we are given a set $\PntSet$ of $n$ points in
$\Re^d$, a query point $\query$ and parameters $k$ and $\eps>0$. We
are interested in the following problems: %\smallskip%
\begin{compactenum}[ I.]
    \item \emphi{\SLR} (nearest induced linear subspace): Compute
    $\nnY{\query}{\bigl.\LFlatsY{k}{\PntSet}}$.

    \item \emphi{\ANLF} ({approximate nearest linear induced flat}):
    Compute a $k$-flat $\flat \in \LFlatsY{k}{\PntSet}$, such that
    $\distSet{\query}{\flat} \leq (1+\eps)
    \distSet{\query}{\bigl.\LFlatsY{k}{\PntSet}}$.

    \item \emphi{Affine \SLR} (nearest induced flat): Compute
    $\nnY{\query}{\bigl.\FlatsY{k}{\PntSet}}$.

    \item \emphi{\ANIF} ({Approximate Nearest Induced Flat}): Compute
    a $(k-1)$-flat $\flat \in \FlatsY{k}{\PntSet}$, such that
    $\distSet{\query}{\flat} \leq (1+\eps)
    \distSet{\query}{\bigl.\FlatsY{k}{\PntSet}}$.

    \item \emphi{Convex \SLR} ({Nearest Induced Simplex}): Compute
    $\nnY{\query}{\bigl.\SimpsY{k}{\PntSet}}$.

    \item \emphi{\ANIS} ({Approximate Nearest Induced Simplex}):
    Compute a $(k-1)$-simplex $\simp \in \SimpsY{k}{\PntSet}$, such
    that
    $\distSet{\query}{\simp} \leq (1+\eps)
    \distSet{\query}{\bigl.\SimpsY{k}{\PntSet}}$.
\end{compactenum}
% \smallskip%
Here, the parameter $k$ corresponds to the \emph{sparsity} of the
solution.

%#############################
%............................... NIF and NLF ...................#
%#############################
%%%%%%%%%%%%%%%%%%%%%%%%%%%%%%%%%%%%%%%%%%%%%%%%%%%%%%%%%%%%%%%%%%%%%
%%%%%%%%%%%%%%%%%%%%%%%%%%%%%%%%%%%%%%%%%%%%%%%%%%%%%%%%%%%%%%%%%%%%%
%%%%%%%%%%%%%%%%%%%%%%%%%%%%%%%%%%%%%%%%%%%%%%%%%%%%%%%%%%%%%%%%%%%%%
\section{Approximating the nearest %
   induced flats and subspaces}
\seclab{sec:nearest-flat}

In this section, we show how to solve approximately the online
variants of \SLR and affine \SLR problems. These results are later
used in \secref{sec:nearest-simplex}. We start with the simplified
case of the uniform star.

\subsection{Approximating the nearest neighbor %
   in a uniform star}

\mypar{Input \& task.} %
We are given a base point $\bpnt$, a set $\PntSet$ of $n$ points in
$\Re^d$, and a parameter $\eps>0$.  We assume that
$\distY{\bpnt}{\pnt} = 1$, for all $\pnt \in \PntSet$.  The task is to
build a data structure that can report quickly, for a query point
$\query$ that is also at distance one from $\bpnt$, the
$(1+\eps)$-\ANN segment to $\query$ in $\starY{\bpnt}{\PntSet}$.

\mypar{Preprocessing.} %
The algorithm computes the set
$\PntSetA = \Set{\pnt - \bpnt}{\pnt \in \PntSet \setminus
   \brc{\bpnt}}$, which lies on a unit sphere in $\Re^{d}$. Next, the
algorithm builds a data structure $\DS_\PntSetA$ for answering
$(1+\eps)$-\ANN queries on $\PntSetA$.

\mypar{Answering a query.} %
For a query point $\query$, the algorithm does the following:
\begin{compactenum}[\quad(A)]
    \item Compute $\vecA = \query - \bpnt$.

    \item Compute $(1+\eps)$-\ANN{} to $\vecA$ in $\PntSetA$, denoted
    by $\vecB$ using $\DS_\PntSetA$.

    \item Let $\pntA$ be the point in $\PntSet$ corresponding to
    $\vecB$.

    \item Return
    $\min\pth{ \bigl. \distSet{\query}{ \bpnt \pntA }, \, 1 }$.
\end{compactenum}

%\subsubsection{Analysis}

\newcommand{\PlaceFigureOnRight}[2]{%
   \noindent%
   \SaveIndent%
   \begin{minipage}{\linewidth - \widthof{#2} - \widthof{\quad}}
       \RestoreIndent%
       #1
   \end{minipage}
   \quad%
   \begin{minipage}{\widthof{#2}}%
       % 0.28\linewidth}
       \hfill%
       #2
   \end{minipage}
}

\begin{lemma}
    \lemlab{ann:u:star}%
    Consider a base point $\bpnt$, and a set $\PntSet$ of $n$ points
    in $\Re^d$ all on $\SphereY{\bpnt}{1}$, where
    $\SphereA =\SphereY{\bpnt}{1}$ is the sphere of radius $1$
    centered at $\bpnt$.  Given a query point $\query \in \SphereA$,
    the above algorithm reports correctly a $(1+\eps)$-\ANN in
    $\starY{\bpnt}{\PntSet}$. The query time is dominated by the time
    to perform a single $(1+\eps)$-\ANN
    query.
\end{lemma}
\begin{proof}
    If $\distSet{\query}{\starY{\bpnt}{\PntSet}} =1$ the correctness
    is obvious.  Otherwise,
    $\distSet{\query}{\starY{\bpnt}{\PntSet}} < 1$, and assume for the
    sake of simplicity of exposition that $\bpnt$ is the origin.  Let
    $\pntE$ be the nearest neighbor to $\query$ in $\PntSet$, and let
    $\pntE'$ be the nearest point on $\bpnt \pntE$ to $\query$.
    Similarly, let $\pntA$ be the point returned by the \ANN data
    structure for $\query$, and let $\pntA'$ be the nearest point on
    $\bpnt \pntA$ to $\query$.  Moreover, let
    $\alpha = \Angle{\query \bpnt \pntE}$ and
    $\beta = \Angle{\query \bpnt \pntA}$. As
    $\distSet{\query}{\starY{\bpnt}{\PntSet}} <1$, we can also
    conclude that $\alpha$ is smaller than $\pi/2$.
    
       We have that
       \begin{math}
           \distY{\query}{\pntE} \leq \distY{\query}{\pntA} \leq
           (1+\eps) \distY{\query}{\pntE}.
       \end{math}
       % Let $\alpha =\Angle{ \query \bpnt \pntE}$ and
       % $\beta =\Angle{ \query \bpnt \pntA}$, and
       Observe that
       \begin{math}
           \distY{\query}{\pntA} = 2 \sin( \beta/2 ),
       \end{math}
       and
       \begin{math}
           \distY{\query}{\pntE} = 2 \sin( \alpha/2 ).
       \end{math}
       As such, by the monotonicity of the $\sin$ function in the
       range $[0,\pi/2]$, we conclude that $\alpha \leq \beta$. This
       readily implies that $\pntE$ is the point of $\PntSet$ that
       minimizes the angle $\measuredangle \query \bpnt \pntE$ (i.e.,
       $\alpha$), which in turn minimizes the distance to $\query$ on
       the induced segment. As such, we have
       \begin{math}
           \distSet{\query}{\starY{\bpnt}{\PntSet}} =
           \distY{\query}{\pntE'}.
       \end{math}
     
\begin{figure}[!h]
\centerline{\includegraphics[width=4cm]{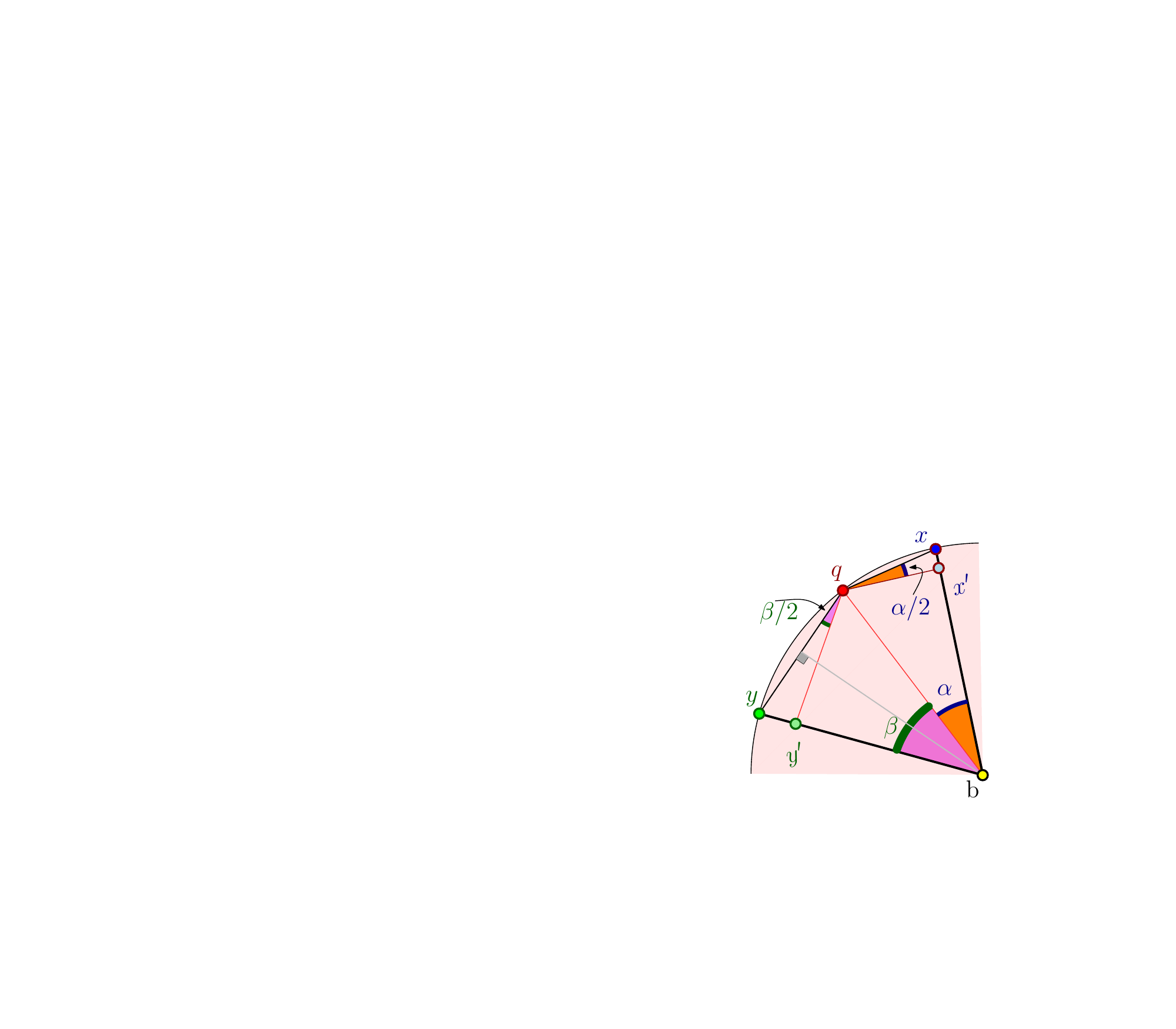}}%
\caption{}
\figlab{lemFigure}
\end{figure}
       
       As for the quality of approximation, first suppose that
       $\beta\leq \pi/2$. Then we have
       \begin{math}
           \frac{\distY{\query}{\pntA'}}{\distY{\query}{\pntE'}}%
           =%
           \frac{\distY{\query}{\pntA}\cos (\beta /2) }%
           {\distY{\query}{\pntE}\cos (\alpha/2)}%
           \leq%
           1+\eps,
       \end{math}
       since
       $\distY{\query}{\pntA}/ \distY{\query}{\pntE} \leq 1+\eps$ and
       $\alpha/2 \leq \beta/2 \leq \pi/2$, which in turn implies that
       $\cos (\beta/2) \leq \cos (\alpha/2)$. %

    Otherwise, we know that $\beta> \pi/2$ and thus
    $\distSet{\query}{\bpnt\pntA} = 1$. Now if
    $1\leq (1+\eps) \distY{\query}{\pntE'}$ we are done. Thus, we
    assume that $\distY{\query}{\pntE'}\leq 1/(1+\eps)$. Also, we have
    that
    $\distY{\query}{\pntE}\geq \distY{\query}{\pntA}/(1+\eps) >
    \sqrt{2}/(1+\eps)$ as $\beta> \pi/2$. This would imply that
    $\sqrt{2}/(1+\eps) <\distY{\query}{\pntE} =
    \distY{\query}{\pntE'}/\cos (\alpha/2) \leq 2/\sqrt{2} (1+\eps)$,
    as $\alpha\leq \pi/2$, which is a contradiction. Hence, the lemma
    holds.  %
\end{proof}

% #####################

\subsection{Approximating the nearest %
   flat in a bouquet}
\seclab{bouquet}

\begin{definition}%[Direction]
    \deflab{dir}%
    For a set $\SetX$ and a point $\pnt$ in $\Re^d$, let
    $\pnt'= \nnY{\pnt}{\SetX}$.  We use $\dirY{\SetX}{\pnt}$ to denote
    the unit vector $(\pnt-\pnt') / \norm{\pnt-\pnt'}$, which is the
    \emphi{direction} of $\pnt$ in relation to $\SetX$.
\end{definition}

\mypar{Input \& task.}  We are given sets $\Pb$ and $\PntSet$ of $k-1$
and $n$ points, respectively, in $\Re^d$, and a parameter $\eps>0$.
The task is to build a data structure that can report quickly, for a
query point $\query$, a $(1+\eps)$-\ANN flat to $\query$ in
$\bouqY{\Pb}{\PntSet}$, see \defref{star}.

\mypar{Preprocessing.} %
Let $\FlatB = \flatX{\Pb}$.  The algorithm computes the set
\begin{math}
    \PntSetA = \Set{\dirY{\FlatB}{\pnt}, - \dirY{\FlatB}{\pnt}}{\pnt
       \in \PntSet \setminus \Pb},
\end{math}
which lies on a $d-k+2$ dimensional unit sphere in $\Re^{d-k+1}$, and
then builds a data structure $\DS_\PntSetA$ for answering (standard)
\ANN queries on $\PntSetA$.

\mypar{Answering a query.} %
For a query point $\query$, the algorithm does the following:
\begin{compactenum}[\quad(A)]
    \item Compute $\vecA = \dirY{\FlatB}{\query}$.
    \item Compute \ANN{} to $\vecA$ in $\PntSetA$, denoted by $\vecB$
    using the data structure $\DS_\PntSetA$.

    \item Let $\pnt$ be the point in $\PntSet$ corresponding to
    $\vecB$.

    \item Return the distance
    $\distSet{\query}{\flatX{\Pb \cup \brc{\pnt}}}$.
\end{compactenum}
% The algorithm returns the minimum distance encountered as the
% desired approximation.

%\subsubsection{Analysis}%

\begin{definition}
    For sets $\SetX, \SetY \subseteq \Re^d$, let
    $\projY{\SetX}{\SetY} = \Set{\nnY{\query}{\SetX} }{\query \in
       \SetY}$ be the \emphi{projection} of $\SetY$ on $\SetX$.
\end{definition}

\begin{lemma}
    \lemlab{project:flat}%
    Consider two affine subspaces $\flatF \subseteq \flatH$ with a
    base point $\bpnt \in \flatF$, and the orthogonal complement
    affine subspace
    \begin{math}
        \ocX{\flatF} = \Set{ \bpnt + \vecA }{ \bigl.%
           \DotProd{\vecA}{\pntB -\pntC} = 0 %
           \text{ for all } \pntB,\pntC \in \flatF , \vecA\in
           \Re^d}\!.
    \end{math}
    For an arbitrary point $\query \in \Re^d$, let
    $\ocX{\query} = \projY{\ocX{\flatF}}{\query}$.% and
    We have that
    \begin{math}
        \distSet{\query}{\flatH} =%
        \distSet{\ocX{\query}}{\projY{\ocX{\flatF}}{\flatH}}.
    \end{math}
\end{lemma}
\begin{proof}
    For the sake of simplicity of exposition, assume that $\bpnt$ is
    the origin.  Let $\flatG = \ocX{\flatF} \cap \flatH$. We have that
    $\flatF$ and $\flatG$ are both contained in $\flatH$, and are
    orthogonal to each other. Furthermore, we have
    $\flatH = \flatF + \flatG$. Setting $f = \dimX{\flatF}$,
    $g = \dimX{\flatG}$, and $h =\dimX{\flatH}$, we have $h = f+g$,
    and
    \begin{math}
        \projY{\ocX{\flatF}}{\flatH}%
        =%
        \projY{\ocX{\flatF}}{\flatF + \flatG}%
        =%
        \projY{\ocX{\flatF}}{\flatF} + \projY{\ocX{\flatF}}{\flatG}%
        =%
        \flatG.
    \end{math}
    
    So, assume that the $\flatF$ spans the first $f$ coordinates of
    $\Re^d$, and $\flatG$ spans the next $g$ coordinates (i.e.,
    $\flatH$ is the linear subspace defined by the first $h$
    coordinates).
    % Let $e_i$ be the $i$\th vector in the
    % orthonormal basis of $\Re^d$.
    If $\query = (q_1, \ldots, q_d)$, then
    $\ocX{\query} = (0, \ldots, 0, q^{}_{f+1}, \ldots, q^{}_d)$, and
    \begin{align*}
      \distSet{\query}{\flatH}%
      =%
      \sqrt{\sum_{i=h+1}^d q_i^2}
      = %
      \sqrt{\sum_{i=1}^f  0^2 +\sum_{i=h+1}^d q_i^2}
      =%
      \distSet{\ocX{\query}}{\flatG}%
      =%
      \distSet{\ocX{\query}}{\projY{\ocX{\flatF}}{\flatG}}.
    \end{align*}
\end{proof}

Using the notation of \asmref{ANN} and \defref{star}, we have the
following:

\begin{lemma}[\ANN flat in a bouquet]
    \lemlab{slr-k:b}%
    Given sets $\Pb$ and $\PntSet$ of $k-1$ and $n$ points,
    respectively, in $\Re^d$, and a parameter $\eps>0$, one can
    preprocess them, using a single \ANN data structure, such that
    given a query point, the algorithm can compute a $(1+\eps)$-\ANN
    to the closest $(k-1)$-flat in $\bouqY{\Pb}{\PntSet}$.  The
    algorithm space and preprocessing time is $O(\spY{n}{\eps})$, and
    the query time is
    $O(\tmQueryY{n}{\eps})$.
\end{lemma}
\begin{proof}
    The construction is described above, and the space and query time
    bounds follow directly from the algorithm description.  As for
    correctness, pick an arbitrary base point $\bpnt \in \Pb$, and let
    $\ocX{\FlatB}$ be the orthogonal complement affine subspace to
    $\FlatB$ passing through $\bpnt$.  Let
    $\ocX{\PntSet} = \projY{\ocX{\FlatB}}{\PntSet}$, and observe that
    $\projY{\ocX{\FlatB}}{\Pb}$ is the point $\bpnt$. In particular,
    the projection of $\bouqA = \bouqY{\Pb}{\PntSet}$ to
    $\ocX{\FlatB}$ is the bouquet of lines
    $\ocX{\bouqA} = \bouqY{\brc{\bpnt}}{\ocX{\PntSet}}$. Applying
    \lemref{project:flat} to each flat of $\bouqA$, implies that
    $\distSet{\query}{\bouqA} = \distSet{\ocX{\query}}{\ocX{\bouqA}}$,
    where $\ocX{\query} = \projY{\ocX{\FlatB}}{\query}$. Let
    $\SphereA$ be the sphere of radius
    $r = \distY{\ocX{\query}}{\bpnt}$ centered at $\bpnt$. Clearly,
    the closest line in the bouquet, is the closest point to the
    uniform star formed by clipping the lines of $\ocX{\bouqA}$ to
    $\SphereA$. Since all the lines of $\ocX{\bouqA}$ pass through
    $\bpnt$, scaling space around $\bpnt$ by some factor $\alpha$,
    just scales the distances between $\query$ and $\ocX{\bouqA}$ by
    $x$. As such, scale space so that $r=1$. But then, this is a
    uniform star with radius $1$, and the algorithm of
    \lemref{ann:u:star} applies (which is exactly what the current
    algorithm is doing). Thus, it correctly identifies a line in the
    bouquet that realizes the desired approximation, implying the
    correctness of the query algorithm.
\end{proof}

% #######################
\subsection{The result}
\seclab{projection} Here, we show simple algorithms for the \ANIF and
the \ANLF problems by employing \lemref{slr-k:b}. We assume $\eps>0$
is a prespecified approximation parameter.

\mypar{Approximating the affine \SLR.}  As discussed earlier, the
goal is to find an approximately closest $(k-1)$-dimensional flat that
passes through $k$ points of $\PntSet$, to the query.  To this end, we
enumerate all possible $k-1$ subsets of points of
$\Pb \subset_{k-1}\PntSet$, and build for each such base set $\Pb$,
the data structure of \lemref{slr-k:b}. Given a query, we compute the
\ANN flat in each one of these data structures, and return the closest
one found.

\begin{theorem}%[Affine \SLR]
    \thmlab{slr-k}%
    %
    % Given a set $\PntSet$ of $n$ points in $\Re^d$, and parameters
    % $k$ and $\eps>0$, one can preprocess them, such that given a
    % query
    % point,
    The aforementioned algorithm computes a $(1+\eps)$-\ANN to the
    closest $(k-1)$-flat in $\FlatsY{k}{\PntSet}$, see
    \defref{induced:flats}.  The space and preprocessing time is
    $O(n^{k-1}\spY{n}{\eps})$, and the query time is
    $O(n^{k-1}\tmQueryY{n}{\eps})$.
\end{theorem}

\mypar{Approximating the \SLR.}  The goal here is to find an
approximately closest $k$-dimensional flat that passes through $k$
points of $\PntSet$ and the origin $\origin$, to the query.  We
enumerate all possible $k-1$ subsets of points of
$\Pb'\subset_{k-1}\PntSet$, and build for each base set
$\Pb=\Pb'\cup\{\origin\}$, the data structure of
\lemref{slr-k:b}. Given a query, we compute the \ANN flat in each one
of these data structures, and return the closest one found.

\begin{theorem}%[\SLR]
    \thmlab{slr-k2}%
    The aforementioned algorithm computes a $(1+\eps)$-\ANN to the
    closest $k$-flat in $\LFlatsY{k}{\PntSet}$, see
    \defref{induced:flats}, with space and preprocessing time of
    $O(n^{k-1}\spY{n}{\eps})$, and the query time of
    $O(n^{k-1}\allowbreak\tmQueryY{n}{\eps})$.
\end{theorem}

%%%%%%%%%%%%%%%%%%%%%%%%%%%%%%%%%%%%%%%%%%%%%%%%%%%%%%%%%%%%%%%%%%%%%%%%
%%%%%%%%%%%%%%%%%%%%%%%%%%%%%%%%%%%%%%%%%%%%%%%%%%%%%%%%%%%%%%%%%%%%%%%%
%%%%%%%%%%%%%%%%%%%%%%%%%%%%%%%%%%%%%%%%%%%%%%%%%%%%%%%%%%%%%%%%%%%%%%%%
\section{Approximating the nearest induced simplex}
% An Algorithm for the online \CSLR problem for $k>2$}%
\seclab{sec:nearest-simplex} In this section we consider the online
variant of the \ANIS problem. Here, we are given the parameter $k$,
and the goal is to build a data structure, such that given a query
point $\query$, it can find a $(1+\eps)$-ANN induced $(k-1)$-simplex.

% {\color{red}
As before, we would like to fix a set $\Pb$ of $k-1$ points and look
for the closest simplex that contains $\Pb$ and an additional point
from $\PntSet$. The plan is to filter out the simplices for which the
projection of the query on to them falls outside of the interior of
the simplex. Then we can use the algorithm of the previous section to
find the closest flat corresponding to the feasible simplices (the
ones that are not filtered out).  First we define a canonical space
and map all these simplices and the query point to a unique
$(k+1)$-dimensional space. As it will become clear shortly, the goal
of this conversion is to have a common lower dimensional space through
which we can find all feasible simplices using range searching
queries.%}

\subsection{Simplices and distances}

\subsubsection{Canonical realization}

In the following, we fix a sequence
$\Pb = \pth{\pnt_1, \ldots, \pnt_{k-1}}$ of $k-1$ points in $\Re^d$.
We are interested in arguing about simplices induced by $k+1$ points,
i.e., $\Pb$, an additional input point $\pnt_k$, and a query point
$\query$. Since the ambient dimension is much higher (i.e., $d$), it
would be useful to have a common canonical space, where we can argue
about all entities.

\begin{definition}
    For a given set of points $\Pb$, let $\FlatB = \flatX{\Pb}$.  Let
    $\pnt \notin \FlatB$ be a given point in $\Re^d$, and consider the
    two connected components of
    $\flatX{\Pb \cup \{\pnt\}} \setminus \FlatB$, which are
    \emphi{halfflats}. The halfflat containing $\pnt$ is the
    \emphi{positive halfflat}, and it is denoted by
    $\pFlatY{\Pb}{\pnt}$.
\end{definition}

Fix some arbitrary point $\pntD^* \in \Re^d \setminus \FlatB$, and let
$\FlatC = \pFlatY{\Pb}{\pntD^*}$ be a \emph{canonical} such
halfflat. Similarly, for a fixed point
$\pntD^{**} \in \Re^d \setminus \flatX{\Pb \cup \brc{\pntD^*}}$, let
$\FlatD = \pFlatY{\Pb \cup \pntD^*}{\pntD^{**}}$. Conceptually, it is
convenient to consider $\FlatD = \Re^{k-2} \times \Re \times \Re^+$,
where the first $k-2$ coordinates correspond to $\FlatB$, and the
first $k-1$ coordinates correspond to $\FlatC$ (this can be done by
applying a translation and a rotation that maps $\FlatD$ into this
desired coordinates system). This is the \emphi{canonical
   parameterization} of $\FlatD$.

% {\color{red}
The following observation formalizes the following: Given a $(k-1)$
dimensional halfflat $\FlatC$ passing through $\Pb$, a point on
$\FlatC$ is uniquely identified by its distances from the points in
$\Pb$.%}

\begin{observation}
    \obslab{obs:pgl} Given a sequence of distances
    $\bm{\vLen} = \pth{ \len_1, \ldots, \len_{k-1}}$, there might be
    only one unique point $\pnt = \pFlat{\vLen} \in \FlatC$, such that
    $\distY{\pnt}{\pnt_i} = \len_i$, for $i=1, \ldots, k-1$. Such a
    point might not exist at all\footnote{\emph{Trilateration} is the
       process of determining the location of $\pnt \in \FlatC$ given
       $\vLen$. \emph{Triangulation} is the process of determining the
       location when one knows the angles (not the distances).}.
\end{observation}

Next, given $\FlatC$ and $\FlatD$, a point $\query$ and a value
$\len < \distSet{\query}{\FlatB}$, we aim to define the points
$\qX{\len}$ and $\qAX{\len}$.  Consider a point
$\query \in \Re^d \setminus \FlatB$ (not necessarily the query point),
and consider any positive $(k-1)$-halfflat $\flatA$ that contains
$\Pb$, and is in distance $\len$ from $\query$. Furthermore assume
that
\begin{math}
    \len = \distSet{\query}{\flatA} < \distSet{\query}{\FlatB}.
\end{math}
Let $\qf$ be the projection of $\query$ to $\flatA$. Observe that, by
the Pythagorean theorem, we have that
$d_i = \distY{\qf}{\pnt_i} = \sqrt{\distY{\query}{\pnt_i}^2 -
   \len^2}$, for $i=1,\ldots, k-1$. Thus, the above observation
implies, that the canonical point
$\qX{\len} = \pFlat{\bigl. d_1, \ldots, d_{k-1}}$ (see
\obsref{obs:pgl}) is uniquely defined. Note that this is somewhat
counterintuitive as the flat $\flatA$ and thus the point $\qf$ are not
uniquely defined. Similarly, there is a unique point
$\qAX{\len} \in \FlatD$, such that:
\begin{inparaenum}[(i)]
    \item the projection of $\qAX{\len}$ to $\FlatC$ is the point
    $\qX{\len}$,
    \item $\distY{\qAX{\len}}{\qX{\len}} = \len$, and these two also
    imply that
    \item $\distY{\qAX{\len}}{\pnt_i} = \distY{\query}{\pnt_i}$, for
    $i=1, \ldots, k-1$.
\end{inparaenum}

Therefore, given $\FlatC$ and $\FlatD$, a point $\query$
and a value $\len < \distSet{\query}{\FlatB}$, the points $\qX{\len}$
and $\qAX{\len}$ are uniquely defined. Intuitively, for a halfflat
that passes through $\Pb$ and is at distance $\len$ from the query,
$\qX{\len}$ models the position of the projection of the query onto
the halfflat, and $\qAX{\len}$ models the position of the query point
itself with respect to this halfflat. Next, we prove certain
properties of these points.

\subsubsection{Orbits}
\begin{definition}
    For a set of points $\Pb$ in $\Re^d$, define $\PrismB$ to be the
    open set of all points in $\Re^d$, such that their projection into
    $\FlatB$ lies in the interior of the simplex
    $\simX{\Pb} = \CHX{\Pb}$. The set $\PrismB$ is a \emphi{prism}.
\end{definition}

Consider a query point $\query \in \PrismB$, and its projection
$\qB = \nnY{\query}{\FlatB}$. Let
$\rad = \radBX{\query} = \distY{\query}{\qB}$ be the \emphi{radius} of
$\query$ in relation to $\Pb$. Using the above canonical
parameterization, we have that $\qX{0} = (\qB,\rad)$, and
$\qAX{0} = (\qX{0}, 0) = (\qB, \rad, 0)$. More generally, for
$\len \in [0,\rad]$, we have
\begin{align}
  \qX{\len} = \pth{\qB, \sqrt{\rad^2 - \len^2}}
  \qquad \text{ and } \qquad%
  \qAX{\len} = \pth{\qB, \sqrt{\rad^2 - \len^2}, \len}.
  \eqlab{q:A:X}
\end{align}
The curve traced by $\qAX{\len}$, as $\len$ varies from $0$ to $\rad$,
is the \emphi{orbit} of $\query$ -- it is a quarter circle with radius
$\rad$.  The following lemma states a monotonicity property that is
the basis for the binary search over the value of $\len$.

\begin{lemma}
    \lemlab{monotone}%
    (i) Define
    $\wquery(\len) = \bigl(\sqrt{\rad^2 - \len^2}, \len \bigr)$, and
    consider any point $\pnt = (x,0)$, where $x \geq 0$. Then, the
    function $\wdX{\len } = \distY{\wquery(\len)}{\pnt}$ is
    monotonically increasing for $\ell \in [0,\rad]$.

    (ii) % Let $\query \in \PrismB$.
    For any point $\pnt$ in the halfflat $\FlatC$, the function
    $\distY{\qAX{\len}}{\pnt}$ is monotonically increasing.
\end{lemma}
\begin{proof}
    (i)
    \begin{math}
        D(\len)%
        =%
        (\wdX{\len })^2%
        =%
        \bigl( \sqrt{\rad^2 - \len^2} - x \bigr)^2 + \len^2%
        =%
        \rad^2 - 2 x \sqrt{\rad^2 - \len^2} + x^2,
    \end{math}
    and clearly this is a monotonically increasing function for
    $\ell \in [0,\rad]$.

    (ii) Let $\pnt_\Pb = \nnY{\pnt}{\FlatB}$, and let $x \geq 0$ be a
    number such that in the canonical representation, we have that
    $\pnt = (\pnt_\Pb, x)$. Using \Eqref{q:A:X}, we have
    \begin{align*}
      D(\len)%
      =%
      % (d(\len))^2%
      % =%
      \distY{\qAX{\len}}{\pnt}^2%
      =%
      \distY{\bigl(\qB, \sqrt{\rad^2 - \len^2}, \len
      \bigr)}{ \bigl( \pnt_\Pb,x,0 \bigr) }^2%      
      =%
      \distY{\qB}{\pnt_\Pb}^2 
      + \pth{\bigl.\wdX{\len}}^2,
    \end{align*}
    and the claim readily follows from (i) as $\qB$ and $\pnt_\Pb$ do
    not depend on $\len$.
\end{proof}
\subsubsection{Distance to a simplex via distance to the flat}

\begin{definition}
    \deflab{p:simplex}%
    Given a point $\query$, and a distance $\len$, let
    $\simCY{\query}{\len}$ be the unique simplex in $\FlatC$, having
    the points of $\Pb$ and the point $\qX{\len}$ as its vertices.
    Similarly, let $\simCX{\query} = \simCY{\query}{0}$.
\end{definition}

\noindent Next, we provide the necessary and sufficient conditions
for a simplex to be feasible. This lemma lies at the heart of our data structure.

\begin{lemma}
    \lemlab{w:nn}%
    Given a query point $\query \in \PrismB$, and a point
    $\pnt_k \in \PntSet \setminus \Pb$, for a number
    $0<x\leq \distSet{\query}{\FlatB}$ we have
    \begin{compactenum}[\quad(A)]
        \item ${\qX{x}} \in \simCX{\pnt_k}$ and
        $\distSet{\query}{\pFlatY{\Pb }{\pnt_k}} \leq x$ $\implies$
        \begin{math}
            \distSet{\query}{\simX{\Pb \cup \brc{\pnt_k}}} \leq x.
        \end{math}

        \item
        \begin{math}
            \distSet{\query}{\simX{\Pb \cup \brc{\pnt_k}}} \leq x
        \end{math}
        and $\query \in \PrismBX{\pnt_k}$ $\implies$
        ${\qX{x}} \in \simCX{\pnt_k}$ and
        $\distSet{\query}{\pFlatY{\Pb }{\pnt_k}} \leq x$.
    \end{compactenum}
\end{lemma}
\begin{proof}
    (A) Let $\nnq = \nnY{\query}{\pFlatY{\Pb }{\pnt_k}}$. If
    $\nnq \in \FlatB$ then $\nnq \in \intX{\simX{\Pb}}$ (see
    \defref{simplex}), because $\query \in \PrismB$. But then,
    $\distSet{\query}{\simX{\Pb \cup \brc{\pnt_k}}} =
    \distY{\query}{\nnq} \leq x$, as desired.

    Observe that by definition, the projection of $\qAX{x}$ to
    $\FlatC$ is the point $\qX{x}$, and since
    ${\qX{x}} \in \simCX{\pnt_k}$, it follows that
    \begin{math}
        \distSet{\qAX{x}}{ \simCX{\pnt_k}}%
        =%
        \distY{\qX{x}}{\qAX{x}}%
        =%
        x.
    \end{math}
    As such, let $\len = \distSet{\query}{\pFlatY{\Pb }{\pnt_k}}$, and
    observe that by definition of the parameterization
    \begin{math}
        \distSet{\query}{\simX{\Pb \cup \brc{\pnt_k}}}%
        =%
        \distSet{\qAX{\len}}{ \simCX{\pnt_k}}%
    \end{math}. However, since $\len \leq x$ and by \lemref{monotone}
    (ii), we have that
    \begin{math}
        \distSet{\qAX{\len}}{ \simCX{\pnt_k}}%
        \leq%
        \distSet{\qAX{x}}{ \simCX{\pnt_k}}%
        =%
        x.
    \end{math}
	
    \smallskip%
    (B) Let
    \begin{math}
        \len%
        =%
        \distSet{\query}{\pFlatY{\Pb }{\pnt_k}}
    \end{math}, and note that since $\query \in \PrismBX{\pnt_k}$, we
    have that
    \begin{math}
        \len%
        =%
        \distSet{\query}{\simX{\Pb \cup \brc{\pnt_k}}}%
        \leq%
        x,
    \end{math}
    proving one part.  Also, by definition of the parameterization, we
    have that
    $\distSet{\qAX{\len}}{\simCX{\pnt_k}} = \distSet{\query}{\simX{\Pb
          \cup \brc{\pnt_k}}}$ which is equal to $\len$.  As by
    definition $\qX{\len}$ is the projection of $\qAX{\len}$ onto
    $\FlatC$ and that $\distY{\qAX{\len}}{\qX{\len}} = \len$. Thus,
    $\qX{\len}$ is equal to $\nnY{\qAX{\len}}{\simCX{\pnt_k}}$, and
    therefore $\qX{\len} \in \simCX{\pnt_k}$.
    
    Moreover, \Eqref{q:A:X} implies that $\qX{y}$, for
    $y \in [0, \radBX{\query}]$, moves continuously and monotonically
    on a straight segment starting at $\qX{0}$, and as $y$ increases,
    it moves towards $\qX{\rad} = \qB =
    \nnY{\query}{\simX{\Pb}}$. Since $\len \leq x$, and by the
    assumption of the lemma, we have
    $x \leq \distSet{\query}{\FlatB} = \rad$. Thus we infer that
    $\qX{x}$ lies on the segment $\qB\qX{\len}$.  As
    $\query \in \PrismB$, we have that $\qB\in \simCX{\pnt_k}$ and
    $\qX{\len} \in \simCX{\pnt_k}$. Thus, by the convexity of
    $\simCX{\pnt_k}$, we have
    $\qX{x} \in\qB\qX{\len} \subseteq \simCX{\pnt_k}$.
\end{proof}

\subsection{Approximating the nearest page in a book}
\seclab{nnpageinbook}

\begin{definition}
    \deflab{book:page}%
    Let $\PntSet$ be a set of $n$ points in $\Re^d$, and let $\Pb$ be
    a sequence of $k-1$ points. Consider the set of simplices having
    $\Pb$ and one additional point from $\PntSet$; that is,
$
      \SimSetC%
      = %
      \SimSetI{\Pb}{\PntSet}%
      =%
      \Set{ \simX{\Pb\cup\brc{\pnt}} }{ \pnt \in \PntSet \setminus \Pb}.
$
    The set $\SimSetC$ is the \emphi{book} induced by $(\Pb,\PntSet)$,
    and to a single simplex in this book is (naturally) a
    \emphi{page}.
\end{definition}

The task at hand, is to preprocess $\SimSetC$ for \ANN queries, as
long as (i) the nearest point lies in the interior of one of these
simplices and (ii) $\query \in \PrismB$. To this end, we consider the
canonical representation of this set of simplices
$
  \SimSetD%
  = %
  \Set{ \simCX{\pnt} }{ \pnt \in \PntSet \setminus \Pb}.
  \eqlab{s:set}%
$

\mypar{Idea.}  The algorithm follows \lemref{w:nn} (A). Given a query
point, using standard range-searching techniques, we extract a small
number of canonical sets of the points, that in the parametric space,
their simplex contains the parameterized query point. This is
described in \secref{r:simplices}. For each of these canonical sets,
we use the data structure of \lemref{slr-k:b} to quickly query each
one of these canonical sets for their nearest positive flat (see
\remref{a:n:n:p:flat} below). This would give us the desired \ANN.

\subsubsection{Reporting all simplices containing a point}
\seclab{r:simplices}

   \begin{definition}
       Let $\Pb = \pth{\pnt_1, \ldots, \pnt_{k-1}}$ be a sequence of
       $k-1$ points in $\Re^d$.  For a point $\pnt \in \Re^d$,
       consider the $(k-1)$-simplex $\simX{\Pb \cup \brc{\pnt}}$,
       which is a full dimensional simplex in the flat
       $\flatX{\Pb \cup \brc{\pnt}}$ (see \defref{simplex}).  The
       \emphi{base angles} of $\pnt$ (with respect to $\Pb$), is the
       $(k-1)$-tuple
       \begin{math}
           \bAngleY{\Pb}{\pnt} = \pth{\bigl.\ba_1(\pnt),\ldots,
              \ba_{k-1}(\pnt)},
       \end{math}
       where $\ba_i(\pnt)$ is the dihedral angle between the facet
       $\simX{\Pb\cup\brc{\pnt}\setminus \brc{\pnt_i}}$ and the base
       facet $\simX{\Pb}$. See the figure on the right, where $k=3$.
   \end{definition}
\vspace{-0.5cm}
\begin{figure}[!h]
\centerline{\includegraphics[width=6cm]{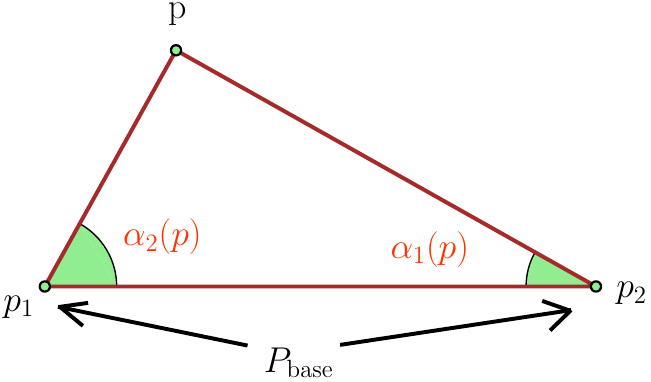}}
\caption{Example  base angles when $k=3$}
\figlab{fig:baseangles}
\end{figure}

% \begin{definition}
%     Let $\Pb = \pth{\pnt_1, \ldots, \pnt_{k-1}}$ be a sequence of
%     $k-1$ points in $\Re^d$.  For a point $\pnt \in \Re^d$, consider
%     the $(k-1)$-simplex $\simX{\Pb \cup \brc{\pnt}}$, which is a full
%     dimensional simplex in the flat $\flatX{\Pb \cup \brc{\pnt}}$ (see
%     \defref{simplex}).  The \emphi{base angles} of $\pnt$ (with
%     respect to $\Pb$), is the $(k-1)$-tuple
%     \begin{math}
%         \bAngleY{\Pb}{\pnt} = \pth{\bigl.\ba_1(\pnt),\ldots,
%            \ba_{k-1}(\pnt)},
%     \end{math}
%     where $\ba_i(\pnt)$ is the dihedral angle between the facet
%     $\simX{\Pb\cup\brc{\pnt}\setminus \brc{\pnt_i}}$ and the base
%     facet $\simX{\Pb}$.  See \figref{fig:baseangles}.
% \end{definition}

% \begin{figure}[!h]
%     \centerline{\includegraphics[width=6cm]{figs/base_angles}}
%     \caption{Example base angles when $k=3$}
%     \figlab{fig:baseangles}
% \end{figure}

\begin{observation}[Inclusion and base angles]
    \obslab{obs:baseangles}%
    Let $\Pb$ be a set of $k-1$ points in $\Re^{k-1}$ all with their
    $(k-1)$\th coordinate being zero, and let $\pnt$ be an additional
    point with its $(k-1)$\th coordinate being a positive
    number. Then, for a point $\query \in \Re^{k-1}$, we have that
    $\query \in \simX{\Pb \cup \brc{\pnt}}$ $\iff$
    $\bAngleY{\Pb}{\query} \leq \bAngleY{\Pb}{\pnt}$ (i.e.,
    $(\forall i: \ba_i(\query)\leq \ba_i(\pnt)$).
\end{observation}

\begin{lemma}
    \lemlab{canonical}%
    Given a set $n$ of $(k-1)$-simplices $\SimSetD$ in $\Re^{k-1}$,
    that all share common $k-1$ vertices, one can build a data
    structure of size $O( n \log^{k-1} n)$, such that given a query
    point $\query\in \Re^{k-1}$, one can compute $O( \log^{k-1} n)$
    disjoint canonical sets, such that the union of these sets, is the
    set of all simplices in $\SimSetD$ that contain $\query$. The
    query time is $O( \log^{k-1} n)$.
\end{lemma}
\begin{proof}
    By a rigid transformation of the space, we can assume that
    $\FlatB$ is the hyperplane $x_{k-1}=0$, and furthermore, all the
    vertices of $\SimSetD$ have $x_{k-1} \geq 0$ (we can handle the
    simplices $x_{k-1} \leq 0$ in a similar separate data
    structure). Let $\PntSet$ be the vertices of $\SimSetD$ not lying
    on $\FlatB$. We generate the corresponding set of base angles
    \begin{math}
        \PASet = \Set{\bAngleY{\Pb}{\pnt} }{\pnt \in \PntSet}.
    \end{math}
    Preprocess this set for orthogonal range searching, say, using
    range-trees \cite{bcko-cgaa-08}. Given a query point $\query$, by
    \obsref{obs:baseangles}, the desired simplices correspond to all
    points in $\pnt \in \PntSet$, such that
    $\bAngleY{\Pb}{\pnt} \geq \bAngleY{\Pb}{\query} $, which is an
    unbounded box query in the range tree of $\PASet$, with the
    aforementioned performance.
\end{proof}

\begin{lemma}
    \lemlab{canonical:range}%
    The data structure of \lemref{canonical} can be used to report all
    simplices that contain a specific point $\pnt$, and do not contain
    another point $\pnt'$, which is vertically above $\pnt$ (i.e., the
    same point with larger $(k-1)$th coordinate).
    % {\color{red}
    This corresponds to $k$ (possibly unbounded) box queries instead
    of quadrant query in the orthogonal data structure. The query time
    and number of canonical sets will be multiplied by at most
    $k$. The space bound remains the same.  Moreover, we ensure these
    set of $k$ boxes are disjoint. 
    % }
\end{lemma}
\begin{proof}
    % {\color{red}
    If $\pnt'$ is vertically above $\pnt$, we have that
    $\bAngleY{\Pb}{\pnt'} \geq \bAngleY{\Pb}{\pnt}$. Let us denote
    $\bAngleY{\Pb}{\pnt}$ by $\bm{\alpha}$ and $\bAngleY{\Pb}{\pnt'}$
    by $\bm{\alpha}'$. For $i=1, \ldots, k-1$, the algorithm would
    issue the box query
    \begin{align*}
      [\alpha_1',\infty],
      [\alpha_2',\infty],
      \cdots,
      [\alpha_{i-1}',\infty],
      [\alpha_i,\alpha_i'],
      [\alpha_{i+1},\infty],
      \cdots,
      [\alpha_{k-1}',\infty]
    \end{align*}
    It is easy to verify that these boxes are disjoint and their union
    will correspond to all simplices that contain $\pnt$ but not
    $\pnt'$. As each of these boxes correspond to polylogarithmic
    number of canonical sets, in total there are $O(k\log^{k-1} n)$
    canonical sets.
\end{proof}

\subsubsection{Data structure and
   correctness}\seclab{sec:page:ds-correct}

\begin{remark}
    \remlab{a:n:n:p:flat}%
    For a set of points $\PntSet$ and a base set $\Pb$, consider the
    set of positive halfspaces (the \emphi{positive bouquet})
    \begin{math}
        \pbouqY{\Pb}{\PntSet}%
        =%
        \Set{ \pFlatY{\Pb}{\pnt} }{ \pnt \in \PntSet \setminus \Pb}.
    \end{math}
    We can preprocess such a set for \ANN queries readily, by using
    the data structure of \lemref{slr-k:b}. The only modification is
    that for every positive flat we assign one vector (in the positive
    direction), instead of two vectors in both directions which we put
    in the data structure of \secref{bouquet}.
\end{remark}

\noindent\textbf{Preprocessing.}  The algorithm computes the set of canonical
simplices $\SimSetD$, see \Eqref{s:set}. Next, the algorithm builds
the data structure of \lemref{canonical} for this set of
simplices. For each canonical set $\PntSetA$ in this data structure,
for the corresponding set of original points, we build the data
structure of \remref{a:n:n:p:flat} to answer \ANN queries on the
positive bouquet $\pbouqY{\Pb}{\PntSetA}$. (Observe that the total
size of these canonical sets is $O(n \log^{k-1} n)$.)

\mypar{Answering a query.}  Given a query point $\query \in \PrismB$,
the algorithm computes its projection $\qB = \nnY{\query}{\FlatB}$,
where $\FlatB = \flatX{\Pb}$.  Let $\rad = \distY{\query}{\qB}$ be the
radius of $\query$. The desired \ANN distance is somewhere in the
interval $[0,\rad]$, and the algorithm maintains an interval
$[\alpha, \beta]$ where this distance lies, and uses binary search to
keep pruning away on this interval, till reaching the desired
approximation.

\newcommand{\critX}[1]{\Mh{\gamma}\pth{#1}}

Observe that for every point $\pnt \in \PntSet$, there is a critical
value $\critX{\pnt}$, such that for $x \geq \critX{\pnt}$, the
parameterized point $\qX{x}$ is inside the simplex $\simCX{\pnt}$, and
is outside if $x < \critX{\pnt}$, see \defref{p:simplex}. Note that
this statement only holds for queries in $\PrismB$ (otherwise it could
have been false on simplices $\simX{\Pb \cup \brc{\pnt}}$ with obtuse
angles, see \remref{obtuse} for handling the case of
$q\notin\PrismB$).

Now, by \lemref{canonical:range}, we can compute a polylogarithmic
number of canonical sets, such that the union of these sets, are
(exactly) all the points with critical values in the range
$[\alpha,\beta)$. As long as the number of critical values is at least
one, we randomly pick one of these values (by sampling from the
canonical sets -- one can assume each canonical set is stored in an
array), and let $\gamma$ be this value. We have to decide if the
desired \ANN is smaller or larger than $\gamma$. To this end, we
compute a representation, by polylogarithmic number of canonical sets,
of all the points of $\PntSet$ such that their simplex contains the
parameterized point $\qX{\gamma}$, using \lemref{canonical}. For each
such canonical set, the algorithm computes the approximate closest
positive halfflat, see \remref{a:n:n:p:flat}. Let $\tau$ be the
minimum distance of such a halfflat computed. If this distance is
smaller than $\gamma$, then the desired \ANN is smaller than $\gamma$,
and the algorithm continues the search in the interval
$[\alpha,\gamma)$, otherwise, the algorithm continues the search in
the interval $[\gamma, \beta)$.

After logarithmic number of steps, in expectation, we have an interval
$[\alpha', \beta')$, that contains no critical value in it, and the
desired \ANN distance lies in this interval. We compute the \ANN
positive flats for all the points that their parameterized simplex
contains $\qX{\beta'}$, and we return this as the desired \ANN
distance.

\mypar{Correctness.} For the sake of simplicity, we first assume the
\ANN data structure returns the exact nearest-neighbor.  \lemref{w:nn}
(A) readily implies that whatever simplex is being returned, its
distance from the query point is the closeset, as claimed by the data
structure. The other direction is more interesting -- consider the
unknown point $\pnt_k$, such that the (desired) nearest point to the
query point lies in the interior of the simplex
$\simX{\Pb \cup \brc{\pnt_k}}$ (we made this assumption at the
beginning of \secref{nnpageinbook}).  \lemref{w:nn} (B) implies that
the distance to this simplex is always going to be inside the active
interval, as the algorithm searches (if not, then the algorithm had
found even closer simplex, which is a contradiction).

\newcommand{\SetA}{\Mh{X}}%

To adapt the proof to the approximate case, suppose that the data
structure of \remref{a:n:n:p:flat} returns a $(1+\eps)$ approximate
nearest neighbor. Consider some iteration of the algorithm.  Let
$\PntSetB \subseteq \PntSet$ be the set of all points $\pnt$ such that
the simplex $\simCX{\pnt}$ contains $\qX{\gamma}$, and let $\SetA$ be
the set of simplices $\simX{\Pb\cup\brc{\pnt}}$ corresponding to the
points in $\PntSetB$, and $\FlatSet$ be the set of half-flats
$\pFlatY{\Pb }{\pnt}$ corresponding to the points in $\PntSetB$.

Suppose that $\ell^*$ is the minimum distance of the half-flats in
$\FlatSet$ to the query, and let $\tau$ be the distance of the
half-flat reported by the \ANN data structure to the query. Thus, we
have $\tau\leq \ell^*(1+\eps)$. Note that, if $\tau < \gamma$, the
optimal distance $\ell^*$ is also less than $\gamma$ and recursing on
the interval $[\alpha,\gamma)$ works as it satisfies the
precondition. However, if $\tau \geq \gamma$, then either
$\ell^* \geq \gamma$ as well, in which case the recursion would work
for the same reason, or $\ell^* <\gamma \leq \tau \leq\ell^*(1+\eps)$.

In the latter case, let $\pnt$ be the reported point corresponding to
$\tau$. We know that the distance of the query to the half-flat
$\pFlatY{\Pb }{\pnt}$ is $\tau$ which is at most
$\ell^*(1+\eps)$. Now, if the simplex $\simCX{\pnt}$ contains the
point $\query_G(\tau)$, as well, then the distance of the query to the
simplex $\simX{\Pb\cup\brc{\pnt}}$ is equal to its distance to the
half-flat $\pFlatY{\Pb }{\pnt}$ which is $\tau \leq \ell^*(1+\eps)$.
Therefore, we can assume that $\query_G(\tau)$ is outside of the
simplex $\simCX{\pnt}$. However in this case, we get that the distance
of the query to the simplex $\simX{\Pb\cup\brc{\pnt}}$ is
\begin{align*}
  \rho
  &
    =%
    \distSet{\query}{\simX{\Pb\cup\brc{\pnt}}}%
    =
    \distSet{\qAX{\tau}}{\simCX{\pnt}}%
    \leq%
    \distSet{\qAX{\tau}}{\qX{\gamma}}%
    =%
    \sqrt{ \tau^2 + \distSet{\query_G(\tau)}{\query_G(\gamma)}^2}
  \\&%
  =%
  \sqrt{ \tau^2 + (\sqrt{\rad^2-\gamma^2} -\sqrt{\rad^2-\tau^2})^2}
  \leq \sqrt{\tau^2 + \tau^2-\gamma^2} \leq \ell^*\sqrt{2(1+\eps)^2
  - 1} < \ell^*(1+2\eps),
\end{align*}
using $\qX{\gamma}\in \simCX{\pnt}$, \Eqref{q:A:X}, and
$\gamma \leq \tau\leq \rad$.  The above implies that the distance of
the query to the simplex $\simX{\Pb\cup\brc{\pnt}}$ is a good
approximation to the distance to the closest simplex. Thus, it is
sufficient to modify the algorithm so that at each iteration, it
checks the distance of simplex it finds to the query, and reports the
best one found in all iterations.

\mypar{Query time.}
In order to sample a value from $[\alpha,\beta)$ the algorithm uses
\lemref{canonical:range} which has running time of $O(k\log^{k-1}
n)$. Moreover, the algorithm performs $O( \log^{k-1} n)$ \ANN queries
in each iteration of the search corresponding to
\lemref{canonical}. As the search takes $O( \log n)$ iterations in
expectation (and also with high-probability), the query time is
$O( \tmQueryY{n}{\eps} \log^ k n)$. Note that this assumes $k$ is
smaller than $\tmQueryY{n}{\eps}$ which holds as we are working in a
low sparsity regime.

\begin{lemma}[Approximate nearest induced page]
    \lemlab{slr-k:page}%
    Given a set $\PntSet$ of $n$ points in $\Re^d$, a set $\Pb$ of
    $k-1$ points, and a parameter $\eps>0$, one can preprocess them,
    such that given a query point, the algorithm computes an
    $(1+\eps)$-\ANN to the closest page in $\SimSetI{\Pb}{\PntSet}$,
    see \defref{book:page}.  This assumes that (i) the nearest point
    to the query lies in the interior of the nearest page, and (ii)
    $\query \in \PrismB$.  The algorithm space and preprocessing time
    is $O(\spY{n}{\eps} \log^k n)$, and the query time is
    $O(\tmQueryY{n}{\eps} \log^k n)$.
\end{lemma}

\subsection{Result: nearest induced simplex}
\seclab{induced-result} The idea is to use brute-force to handle the
distance of the query to the $\leq (k-2)$-simplices induced by the
given point set which takes $O(n^{k-1})$ time.  As such, the remaining
task is to handle the $(k-1)$-simplices, and thus we can assume that
the nearest point to the query lies in the interior of the nearest
simplex,
% (i.e., $\nnY{\query}{\bigl.\LFlatsY{k}{\PntSet}}$)
as desired by \lemref{slr-k:page}.  To this end, we generate the
$\binom{n}{k-1} = O(n^{k-1})$ choices for $\Pb \subseteq \PntSet$, and
for each one of them we build the data structure of
\lemref{slr-k:page}, and query each one of them, returning the closet
one found.

\begin{remark}\remlab{obtuse}
    Note that for a set of $k$ points $\SetA \subset_k \PntSet$, if
    the projection of the query onto the simplex $\simX{A}$ falls
    inside the simplex, i.e. $\query \in \Mh{\Phi}_{A}$, then there
    exists a subset of $k-1$ points $\Pb \subset_{k-1} \SetA$ such
    that the projection of the query onto the simplex $\simX{\Pb}$
    falls inside the simplex, i.e., $\query \in
    \Mh{\Phi}_{\Pb}$. Therefore, either the brute-force component of
    the algorithm finds an \ANN, or there exists a set $\Pb$ for which
    the corresponding data structure reports the correct \ANN.
\end{remark}

We thus get the following result.

\begin{theorem}[Convex SLR]
    \lemlab{slr-k:simplex}%
    Given a set $\PntSet$ of $n$ points in $\Re^d$, and parameters $k$
    and $\eps>0$, one can preprocess them, such that given a query
    point, the algorithm can compute a $(1+\eps)$-\ANN to the closest
    $(k-1)$-simplex in $\SimpsY{k}{\PntSet}$, see
    \defref{induced:flats}.  The algorithm space and preprocessing
    time is $O(n^{k-1}\spY{n}{\eps} \log^{k} n )$, and the query time
    is $O(n^{k-1} \tmQueryY{n}{\eps}\log^{k} n)$.
\end{theorem}

%%%%%%%%%%%%%%%%%%%%%%%%%%%%%%%%%%%%%%%%%%%%%%%%%%%%%%%%%%%%%%%%%%%%%%%%
%%%%%%%%%%%%%%%%%%%%%%%%%%%%%%%%%%%%%%%%%%%%%%%%%%%%%%%%%%%%%%%%%%%%%%%%
%%%%%%%%%%%%%%%%%%%%%%%%%%%%%%%%%%%%%%%%%%%%%%%%%%%%%%%%%%%%%%%%%%%%%%%%

\section{Offline nearest induced segment problem}
\seclab{offline}

In this section, we consider the offline variant of the convex \SLR
problem, for the case of $k=2$. We present an algorithm which achieves
constant factor approximation and has sub-quadratic running time.

   \mypar{Input \& task.} %
   We are given a set of $n$ points
   $\PntSet = \brc{\pnt_1, \cdots, \pnt_n}$, along with the query
   point $\query$, and a parameter $\eps > 0$. The task is to find the
   (approximate) closest segment $\pnt_i\pnt_j$ to the query point
   $\query$, where $\pnt_i,\pnt_j\in \PntSet$.

\mypar{Notation.}%
   Fix a point $\query$, and let
   $r > 0$ be a parameter.  For the sphere
   $\SphereChar = \SphereY{\query}{r}$, and a point $\pntA \in \Re^d$,
   let $\projSX{\pntA}$ be the intersection point of
   $\LineY{\pntA}{\query}$ with the sphere $\SphereChar$ that is
   closer to $\pntA$. The other intersection point is the
   \emph{spherical reflection} of $\pntA$, and is denoted by
   $\refX{\pntA} = \projSX{\pth{\query-(\pntA -\query)}}$.
   % For a
   % point $\pntA \in \Re^d$,

\begin{figure}[!h]
\centerline{\includegraphics[scale=1,page=1]{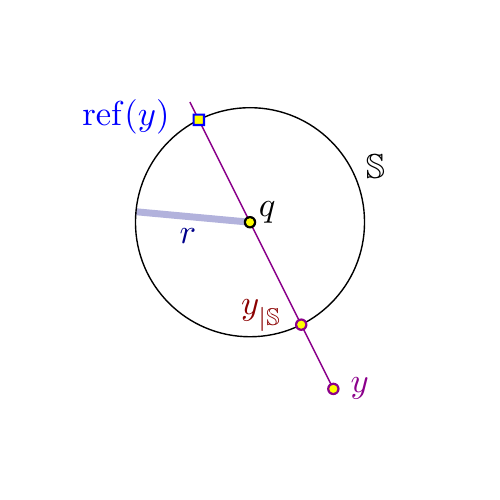}}
\caption{}
\figlab{defs:silly:1}
\end{figure}

   \mypar{Algorithm.} %
The algorithm picks a value $r > 0$ (any value works), and let
$\SphereChar = \SphereY{\query}{r}$. Let
\begin{math}
    \PntSetA%
    =%
    \Setw{ \projSX{\pnt} }{ \pnt \in \PntSet}
\end{math}
be the spherical projection of $\PntSet$ on $\SphereChar$.  Next, the
algorithm builds a data structure $\DS$ for answering $(1+\eps)$-\ANN
queries on $\PntSetA$.

   Now, the algorithm queries this data structure $n$ times. For each
   point $\pnt_i\in \PntSet$, let
\begin{math}
    \pnt_i' =\refX{\pnt_i}
\end{math}
be the spherical reflection of $\pnt_i$.  The algorithm queries the
data structure $\DS$ for the point $\pnt_i'$, and let
$\pntC_i \in \PntSetA$ be the approximate nearest neighbor of
$\pnt_i'$. Let $\widehat{\pnt_i}$ be the original point of $\PntSet$
that induced the point $\pntC_i$.  The algorithm then computes the
distance of the query $\query$ to the segment $\pnt_i\widehat{\pnt_i}$
to see if it improves the closest segment found so far.

\begin{figure}[!h]
\centerline{\includegraphics[scale=1,page=1]{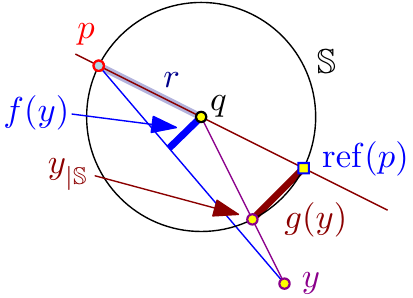}}
\caption{}
\figlab{defs:silly}
\end{figure}

%\newpage

\subsection{Correctness}

   \noindent%
   For points $\pnt, \pntA \in \Re^d$, let $r = \distY{\query}{\pnt}$,
   and let
   \begin{align*}
     f(\pntA)%
     =%
     \distSet{ \query }{ \pnt\pntA}%
     \qquad\text{ and  }\qquad%
     g(\pntA)%
     =%
     \distwY{ \refX{\pnt}}{ \projSX{\pntA}},
   \end{align*}
   see \figref{defs:silly} (note, that $\SphereChar$ is determined by
   the value of $r$ specified).  The task at hand is to approximate
   the minimum of $f(\pntA)$, over all $\pnt, \pntA \in \PntSet$. The
   key observation is that $g(\pntA)$ is a good approximation to
   $f(\pntA)$.
   
\begin{lemma}
    \lemlab{f:g}%
    Using the above notations, consider a point $\pntA \in \Re^d$,
    such that $\distY{\pntA}{\query} \geq r$. Then,
    $f(\pntA) \leq g(\pntA) \leq 2f(\pntA)$.
\end{lemma}

\begin{figure}[!h]%{0.27\linewidth}
    \centerline{%
       \begin{tabular}{cc}
         \includegraphics[scale=0.999,page=3]{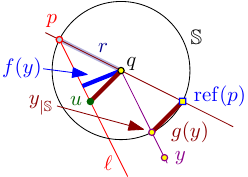}
         \qquad&
                 \qquad
                 \includegraphics[scale=0.999,page=3]{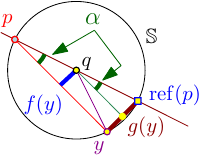}\\
         (A) & (B)            
       \end{tabular}
    } \captionof{figure}{}%
    \figlab{in:proof}%
\end{figure}
\begin{proof}
    As $\pntA$ moves away from $\query$ along $\LineY{\query}{\pntA}$,
    the quantity $g(\pntA)$ remains the same, while $f(\pntA)$
    increases.

    In particular, let $\ell$ be the limit line of
    $\LineY{\pnt}{\pntA}$, as $\pntA$ moves to infinity. Let $\pntB$
    be a point in distance $r$ along $\ell$ from $\pnt$, see
    \figref{in:proof} (A).  Clearly, the two triangles
    $\triangle \pnt \query \pntB$ and
    $\triangle \query \refX{\pnt}\projSX{\pntA}$ are the same up to
    translation. As such, $f(\pntA)$ is bounded by the height of
    $\query$ in $\triangle \pnt \query \pntB$, as this is the limit
    value of $f(\pntA)$ as $\pntA$ moves to infinity.  Thus,
    $f(\pntA)$ is bounded by the length of the edge $\pntB\query$,
    which is equal to $g(\pntA)$. Implying that
    $f(\pntA) \leq g(\pntA)$.

    Similarly, $f(\pntA)$ is minimized when
    $\pntA \in \SphereY{\query}{r}$; that is,
    $\distY{\query}{\pntA} = r$, see \figref{in:proof} (B). Letting
    $\alpha = \measuredangle \query \pnt \pntA$, we have that
    $f(\pntA) = r \sin \alpha$ and
    \begin{math}
        g(\pntA) = 2r \sin (\alpha) = 2f(\pntA),
    \end{math}
    which implies the claim.
\end{proof}

\begin{lemma}
    The above algorithm $2(1+\eps)$-approximates the closest induced
    segment of $\PntSet$.
\end{lemma}
\begin{proof}
    Let $\pnt \pntA$ be the two points such that the segment
    $\pnt \pntA$ is the closest to $\query$ among all the induced
    segments of $\PntSet$. Assume that
    $\distY{\query}{\pnt} \leq \distY{\query}{\pntA}$. Note, that the
    algorithm works the same for any value of $r$. Indeed, the point
    set $\PntSetA$ scales up with the value of $r$, but the
    nearest-neighbor queries are also scaled accordingly, and the
    answers returned by the \ANN data-structure are the same.
    
    In particular, for the analysis, we set $r
    =\distY{\query}{\pnt}$. Assume $\DS$ returned the point
    $\projSX{\pntB}$, and as such, the data-structure computed the
    distance $\Delta = \distSet{\query}{\pnt \pntB} = f(\pntB)$, while
    the desired distance is
    $\delta = \distSet{\query}{\pnt \pntA} = f(\pntA)$.  If
    \begin{math}
        \delta \geq r/[2(1+\eps)],
    \end{math}
    then we are done as the algorithm always returns a distance
    $\leq r$ (since the point $\pnt$ is in distance $r$ from
    $\query$), and thus get the desired approximation. As such, assume
    that $\delta < r/[2(1+\eps)]$.  By \lemref{f:g}, we have that
    $\delta = f( \pntA) \leq g(\pntA) \leq 2f(\pntA) < r/(1+\eps)$. As
    $g(\pntA) = \distwY{\refX{\pnt}}{\projSX{\pntA}}$, and since we
    are using $(1+\eps)$-\ANN data-structure, we have that
    \begin{math}
        g(\pntB)%
        =%
        \distwY{\refX{\pnt}}{\projSX{\pntB}}%
        \leq%
        (1 + \eps)g(\pntA)%
        \leq %
        r.
    \end{math}
    
    If $\distY{\query}{\pntB} \geq r$, then \lemref{f:g} readily
    implies that
    \begin{math}
        f(\pntB) \leq g(\pntB)%
        \leq %
        (1 + \eps)g(\pntA)%
        \leq 2(1 + \eps)f(\pntA),
    \end{math}
    which implies the claim.
    
    If $\distY{\query}{\pntB} < r$, then observe that
    $f(\pntB) \leq f( \projSX{\pntB})$, as moving $u$ in the direction
    of vector from $\pntB$ to $\projSX{\pntB}$ only increases the
    distance of the segment $\pnt \pntB$ from $\query$. As such, by
    the same argument as above, we have that
    \begin{math}
        f(\pntB)%
        \leq%
        f(\projSX{\pntB})%
        \leq%
        g(\pntB)%
        \leq %
        (1 + \eps)g(\pntA)%
        \leq 2(1 + \eps)f(\pntA),
    \end{math}
    as desired.
\end{proof}

We thus get the following.

\begin{theorem}
    \thmlab{n:i:segment}%
    Given a set $\PntSet$ of $n$ points in $\Re^d$, a query point
    $\query$, and a parameter $\eps > 0$, one can compute
    $2(1+\eps)$-\ANN to the closest induced segment by $\PntSet$. The
    algorithm uses $O(\spY{n}{\eps})$ space, and its running time is
    $O\pth{n\tmQueryY{n}{\eps} +\spY{n}{\eps}}$, where $\spY{n}{\eps}$
    is the space and $\tmQueryY{n}{\eps}$ is the query time for a
    $(1+\eps)$-\ANN data-structure for $n$ points in $\Re^d$.
    
    In particular, using the \ANN data-structure of \cite{ai-nohaa-06}
    for the \ANN in the Euclidean space, the resulting running time is
    $n^{1+O({1}/{(1+\eps)^2})}$ (i.e., sub-quadratic in $n$).
\end{theorem}

\begin{corollary}
    \corlab{n:i:s:low:dim}%
    In constant dimension, using the \ANN data-structure of Arya \etal
    \cite{amnsw-oaann-98}, \thmref{n:i:segment} yields a
    data-structure with $O(n )$ preprocessing, and
    $O( n \log n + n/\eps^d)$ running time.
\end{corollary}

\section{Conditional lower bound}
\seclab{lowerbound}

In this section, we reduce the $k$-sum problem to offline variant of
all of our problems, \ANLF, \ANIF and \ANIS problems, providing an
evidence that the time needed to solve these problems is
$\tilde\Omega(n^{k/2}/e^k)$.  In the $k$-sum problem, we are given $n$
integer numbers $a_1, \cdots, a_n$. The goal is to determine if there
exist $k$ numbers among them $a_{i_1},\cdots, a_{i_k}$ such that their
sum equals zero.  The problem is conjectured to require
$\Omega(n^{\lceil k/2 \rceil}/\log^{\Theta(1)} n)$ time, see
\cite{pw-pfsa-10}, Section 5.

We reduce this problem as follows. Let
$\PntSet = \brc{v_1, \cdots, v_n}$ be a set of $n$ vectors of
dimension $k+1$. More precisely, each $v_i \in \Re^{k+1}$ has its
first coordinate equal to $a_i$ and all the other coordinates are $0$
except for one coordinate chosen uniformly at random from
$\{2, \ldots, k+1\}$, whose value we set to $1$.  The query is also a
vector of dimension $(k+1)$ and is of the form
$\query=[0, 1/k, 1/k, \cdots, 1/k]^T$. We query the point $\query$ and
let $v_{i_1},\cdots,v_{i_k}$ be the points corresponding to the
approximate closest flat/simplex reported by the algorithm. We then
check if $\sum_{j=1}^k a_{i_j} = 0$, and if so, we report
$\brc{a_{i_1},\cdots, a_{i_k}}$. Otherwise, we report that no such $k$
numbers exist. Next, we prove the correctness via the following two
lemmas.

\begin{lemma}
    If there is no solution to the $k$-sum problem, that is if there
    is no set $\brc{i_1, \cdots, i_k}$ such that
    $\sum_{j=1}^k a_{i_j} = 0$, then the distance of the query to the
    closest flat/simplex is non-zero.
\end{lemma}
\begin{proof}
    Suppose that the distance of the query to the closest flat/simplex
    is zero. Thus, there exist $k$ vectors $v_{i_1},\cdots, v_{i_k}$
    and the coefficients $c_1,\cdots,c_k$, such that
    $c_1 v_{i_1} + \cdots + c_k v_{i_k} = q$.  Let $t_1, \cdots, t_k$
    be the coordinates ($t_i$ has value from $2$ to $(k+1)$) such that
    $v_{i_j}$ has nonzero value in its $t_j$th coordinate. Note that
    $\query$ has exactly $k$ nonzero coordinates, and each $v_{i_j}$
    has exactly one non-zero coordinate from $2$ to $(k+1)$th
    coordinates, and there are $k$ such vectors. Thus
    $t_1, \cdots, t_k$ should be a permutation from $2$ to
    $k+1$. Therefore all $c_j$'s should be equal to $1/k$. Hence, we
    also have that $a_{i_1}+\cdots + a_{i_k} = 0$ which is a
    contradiction. Therefore the lemma holds.
\end{proof}

\begin{lemma}
    If there exist $a_{i_1}, \cdots, a_{i_k}$ such that
    $\sum_{j=1}^k a_{i_j} = 0$, then with probability $e^{-k}$ the
    solution of the affine \SLR/convex \SLR/\SLR would be a
    flat/simplex which contains the query $\query$.
\end{lemma}
\begin{proof}
    We consider the sum of
    $\frac{1}{k} v_{i_1}+\cdots + \frac{1}{k}v_{i_k}$ and show that it
    equals $\query$ with probability $e^{-k}$. Let $t_j$ be the
    position (from $2$ to $(k+1))$ of the coordinate with value $1$ in
    vector $v_{i_j}$. Then if $t_1, \cdots, t_k$ is a permutation from
    $2$ to $k+1$, we have that
    $\sum_{j=1}^k \frac{1}{k}\cdot v_{i_j} = q$ and thus the solution
    of the affine \SLR/convex \SLR/\SLR would contain $\query$ (notice
    that the coefficients are positive and they sum to $1$, so they
    satisfy the required constraints). The probability that
    $t_1, \cdots, t_k$ is a permutation is
    ${k! \over k^k} \approx {\sqrt{2\pi k}(k/e)^k \over k^k } \geq
    e^{-k}$.
\end{proof}
Therefore we repeat this process $e^k$ times, and if any of the
reported flats/simplices contained the query point, then we report the
corresponding solution $\brc{a_{i_1},\cdots,a_{i_k}}$. Otherwise, we
report that no such $a_{i_1},\cdots, a_{i_k}$ exists. This algorithm
reports the answer correctly with constant probability by the above
lemmas. Moreover as the algorithm needs to detect the case when the
optimal distance is zero or not, this lower bound works for any
approximation of the problem.  Thus we get the following theorem.

\begin{theorem}
    There is an algorithm for the $k$-sum problem with the running
    time bounded by $O(e^{k})$ times the required time to solve any of
    the three variants of the approximate \SLR problem.
    % problem is $\Omega(n^{k/2}/e^k)$ assuming that $k$-sum required
    % $\Omega(n^{k/2})$ time.
\end{theorem}

%%%%%%%%%%%%%%%%%%%%%%%%%%%%%%%%%%%%%%%%%%%%%%%%%%%%%%%%%%%%%%%%%%%%%%%% 
%%%%%%%%%%%%%%%%%%%%%%%%%%%%%%%%%%%%%%%%%%%%%%%%%%%%%%%%%%%%%%%%%%%%%%%%
%%%%%%%%%%%%%%%%%%%%%%%%%%%%%%%%%%%%%%%%%%%%%%%%%%%%%%%%%%%%%%%%%%%%%%%%
\section{Approximating the nearest induced %
   segment}
\seclab{sec:seg}

In this section, we consider the online / query variant of the \ANIS
problem for the case of $k=2$. That is, given a set of $n$ points
$\PntSet \subset \Re^d$, the goal is to preprocess $\PntSet$, such
that given a query point $\query$, one can approximate the closest
segment formed by any of the two points in $\PntSet$ to $\query$.

\subsection{Approximating the nearest neighbor %
   in a star}
\seclab{segment2}

Here, we show how to handle the non-uniform star case -- namely, we
have a set of $n$ segments (of arbitrary length), all sharing an
endpoint, and given a query point, we would like to compute quickly
the \ANN on the star to this query.

\subsubsection{Preliminaries}
\label{sec:segment:prelim}

\begin{lemma}
    \lemlab{prefix}%
    Let $\PntSet = \brc{\pnt_1, \ldots, \pnt_n}$ be an ordered set of
    $n$ points in $\Re^d$. Given a data structure that can answer
    $(1+\eps)$-\ANN queries for $n$ points in $\Re^d$, with space
    $\spY{n}{\eps}$ and query time $\tmQueryY{n}{\eps}$, then one can
    build a data structure, such that for any integer $i$ and a query
    $\query$, one can answer $(1+\eps)$-\ANN queries on the prefix set
    $\PntSet_{i} = \brc{\pnt_1, \ldots, \pnt_i}$. The query time and
    space requirement becomes $O(\spY{n}{\eps} \log n)$ and
    $O(\tmQueryY{n}{\eps} \log n)$, respectively.
\end{lemma}

\begin{proof}
    This is a standard technique used for example in building range
    trees \cite{bcko-cgaa-08}. Build a balanced binary tree over
    $\PntSet$, with the leafs of the tree ordered in the same way as
    in $\PntSet$. Build for any internal node in this tree, the \ANN
    data structure for the canonical set of points stored in this
    subtree. Now, an $(i,\query)$-\ANN query, can be decomposed into
    \ANN queries over $O( \log n)$ such canonical sets. For each
    canonical set, the algorithm uses the \ANN data structure built
    for it. The algorithm returns the best \ANN found.
\end{proof}

Given a set of points $\PntSet$, and a base point $\cen$, its
\emphi{star} is the collection of segments
$\starY{\cen}{\PntSet} = \bigcup^{}_{\pntA \in \PntSet} \cen \pntA$ as
defined in \defref{star}.  The set of points in such a star at
distance $r$ from $\cen$, is the set of points
\begin{align}
  \PRZ{\cen}{r}{\PntSet}%
  =%
  \Set{\pntC}{ \Bigl. \pntA \in \PntSet,
  \normY{\cen}{\pntA} \geq r, \text{ and } \pntC = \cen\pntA \cap
  \SphereY{\cen}{r} },
  \eqlab{star:sphere}
\end{align}
where $\SphereY{\cen}{r}$ denotes the sphere of radius $r$ centered at
$\cen$.

\begin{lemma}
    \lemlab{query:around}%
    Let $\cen$ be a point, and let $\PntSet$ be a set of $n$ points in
    $\Re^d$.  Given a data structure that can answer $(1+\eps)$-\ANN
    queries for $n$ points in $\Re^d$, with space $\spY{n}{\eps}$ and
    query time $\tmQueryY{n}{\eps}$, then one can build a data
    structure, such that given a query point $\query$ and radius $r$,
    one can answer $(1+\eps)$-\ANN queries on the set of points
    $\PRZ{\cen}{r}{\PntSet}$ (see \Eqref{star:sphere}).  The query
    time and space bounds are $O(\spY{n}{\eps} \log n)$ and
    $O(\tmQueryY{n}{\eps} \log n)$, respectively.
\end{lemma}

\begin{proof}
    Let $\PntSet = \brc{\pnt_1, \ldots, \pnt_n}$ be the ordering of
    $\PntSet$ such that the points are in \emph{decreasing} distance
    from $\cen$. Let $\PntSetA = \brc{ \pntC_1, \ldots, \pntC_n}$, be
    the ordering of the direction vectors of the points of $\PntSet$,
    where $\pntC_i = (\pnt_i - \cen) / \norm{ \pnt_i -\cen}$.  Build
    the data structure of \lemref{prefix} on the (ordered) point set
    $\PntSetA$.

    Given a query point $\query$, using a balanced binary search tree,
    find the maximal $i$, such that
    $\ell = \distY{\cen}{\pnt_i} \geq r$. Compute the affine
    transformation $\TrX{\query} = (\query - \cen)/r$, and compute the
    $(1+\eps)$-\ANN to $\TrX{\query}$ in $\PntSetA_i$, and let
    $\pntC_j$ be this point. The algorithm returns $\pnt_j$ is the
    desired \ANN.

    % \smallskip%
    To see why this procedure is correct, observe that
    $\PRZ{\cen}{r}{\PntSet_i} = \PRZ{\cen}{r}{\PntSet}$, where
    $\PntSet_i = \brc{\pnt_1, \ldots, \pnt_i}$.  Furthermore,
    $\PntSetA_i= \brc{ \pntC_1, \ldots, \pntC_i} =
    \TrX{\PRZ{\cen}{r}{\PntSet_i}}$. In particular, since $\TrC$ is
    only translation and scaling, it preserves order between distances
    of pairs of points. In particular, if $\pntA$ is the \ANN to
    $\query$ in $\PRZ{\cen}{r}{\PntSet}$, then $\TrX{\pntA}$ is an
    \ANN to $\TrX{\query}$ in
    $\TrX{\PRZ{\cen}{r}{\PntSet}} = \PntSetA_i$, implying the
    correctness of the above.
\end{proof}

\subsubsection{The query algorithm}
\seclab{alg:detail}%

\mypar{The algorithm in detail.}

We are given a set of $n$ points $\PntSet$ and a center point
$\cen \in \PntSet$, and we are interested in answering $(1+\eps)$-\ANN
queries on $\Star = \starY{\cen}{\PntSet}$.

\mypar{Preprocessing.}  We sort the points of $\PntSet$ to be in
decreasing distance from $\cen$, and let
$\PntSet = \brc{\pnt_1,\ldots, \pnt_n}$ be the points in this sorted
order. We build the data structure $\DS$ of \lemref{query:around} for
the points of $\PntSet$ to answer $(1+\eps/4)$-\ANN queries.

\mypar{Answering a query.}  Given a query point $\query$, let
$r = \distY{\query}{\cen}$.  As a first step, we perform
$(1+\eps/4)$-\ANN query on $\PntSet$. Next, let
$r_i = i (\eps^2/16) r$, for $i=1, \ldots, N = 32 /\eps^2$.  For each
$i$, find the $(1+\eps/4)$-\ANN in the point set
$\PRZ{\cen}{r_i}{\PntSet}$, using the data structure \DS. Return the
nearest point to $\query$ found as the desired \ANN.

\subsubsection{Correctness}
Observe that $\distSet{\query}{\Star} \leq r$, since
$\cen \in \PntSet \subseteq \Star = \starY{\cen}{\PntSet}$.
\begin{lemma}
    \lemlab{far}%
    If $\distSet{\query}{\Star} \geq r \sqrt{\eps}/2$ then the above
    algorithm returns a point $\pnt \in \Star$, such that
    $\distY{\query}{\pnt} \leq (1+\eps) \distSet{\query}{\Star}$, for
    any $0 < \eps \leq 1$.
\end{lemma}
\begin{proof}
    Let $\ell = \distSet{\query}{\Star}$.  Consider the point set
    $\PntSetB_i = \bigcup_{i=1}^N\PRZ{\cen}{r_i}{\PntSet} $. The
    algorithm effectively performs $(1+\eps /4)$-\ANN query over this
    point set. So, let $\query' = \nnY{\query}{\Star}$, and let
    $\pntA \in \PntSet$ be the point, such that
    $\query'\in \cen \pntA$.

    % If $\distY{\query}{\query'} \geq \distY{\query'}{\cen}$ then
    % $\distY{\query}{\query'} \geq r/2$. But then,

    By construction, there is a choice of $i$, and a point
    $\pntB \in \PntSetB_i \cap \cen \pntA$, such that
    $\distY{\pntB}{\query'} \leq (\eps^2/16)r$.  Namely, we have
    $\distSet{\query}{\PntSetB} \leq \distSet{\query}{\Star} +
    (\eps^2/16)r$.  As such, the $(1+\eps/4)$-\ANN point returned for
    $\query$ in $\PntSetB$, is in distance at most
    \begin{align*}
      (1+\eps/4) \distY{\query}{\pntB}%
      &\leq%
        (1+\eps/4) \pth{ \bigl. \distY{\query}{\query'} +
        (\eps^2/16)r}%
        \leq%
        (1+\eps/4) \distY{\query}{\query'} +
        (\eps^2/8)r
      \\&%
      \leq 
      (1+\eps/2) \distSet{\query}{\Star},
    \end{align*}
    since $\distSet{\query}{\cen \pntA} = \distSet{\query}{\Star}$.
\end{proof}

The above lemma implies that (conceptually) the hard case is when the
\ANN distance is small (i.e., $O(\sqrt{\eps} r)$). The intuition is
that in this case the (regular) \ANN query on $\PntSet$ and
$\PRZ{\cen}{r}{\PntSet}$ would ``capture'' this distance, and returns
to us the correct \ANN.  The following somewhat tedious lemma
testifies to this effect.

\begin{lemma}
    \lemlab{choose}%
    Let $\cen, \query$ be two given points, where
    $r=\distY{\query}{\cen}$, and let $0<\eps \leq 1$ be a
    parameter. Consider a set of $n$ points
    $\PntSetC = \brc{\pnt_1, \ldots, \pnt_n} \subseteq
    \ballY{\cen}{r}$, where $\ballY{\cen}{r}$ denotes the ball
    centered at $\cen$ of radius $r$.  Let
    $\Star = \starY{\cen}{\PntSetC}$, and assume that
    $\distSet{\query}{\Star} \leq r \sqrt{\eps}/2$.  Then, we have
    that
    \begin{math}
        \distSet{\query}{\PntSetC}%
        \leq%
        (1+\eps/4)\distSet{\query}{\Star}.
    \end{math}
\end{lemma}

\begin{proof}
    Observe that by definition
    $\distSet{\query}{\Star} \leq \distSet{\query}{\PntSetC}$.  Let
    $\query' = \nnY{\query}{\Star}$, and consider the clipped cone
    $\Cone$ of all points $\pntA\in\PntSetC$ in distance at most $r$
    from $\cen$, such that
    $\Angle{\pntA\ts \cen\ts \query} \leq \phi = \sqrt{\eps}/2$.
    
    Assume that $\query'$ is in $\Cone$. Then, there must be a point
    $\pntA \in \PntSetC \cap \Cone$ such that
    $\query' \in \cen \pntA$.  Observe that
    $\distSet{\query}{\Star} = \distY{\query}{\query'} =
    \distSet{\query}{\cen\pntA}$.  If
    $\distSet{\query}{\cen\pntA} = \distY{\query}{\pntA} \geq
    \distSet{\query}{\PntSetC}$, which implies that
    $\distSet{\query}{\PntSetC} = \distSet{\query}{\Star}$, and we are
    done.

    As such, consider the case that
    $\distSet{\query}{\cen\pntA} < \distY{\query}{\pntA}$, which is
    depicted in \figref{case}. Clearly,
    $\beta = \Angle{\query \pntA \cen}$ is minimized when
    $\distY{\cen }{\pntA} = r$, and
    $\Angle{\pntA \cen \query} = \phi$. But then, this angle is equal
    to $\pi/2 - \phi/2$, which implies that
    $\beta \geq \pi/2 - \phi/2$.

    \begin{figure}[h]
        \centerline{\includegraphics{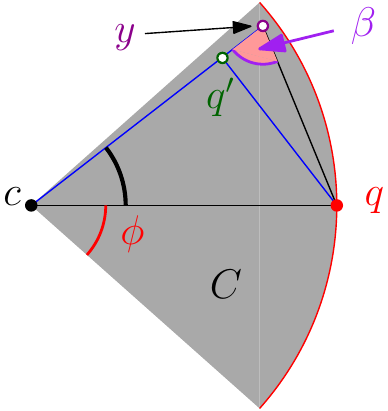}}%
        \caption{Illustration.}
        \figlab{case}
    \end{figure}

    Now, $\normY{\query}{\query'} =\normY{\query}{\pntA}\sin \beta$,
    and $\normY{\query}{\pntA} \leq 2 r \sin (\phi/2)$.  As such,
    \begin{align*}
      \Delta%
      =%
      \normY{\query}{\pntA} - \normY{\query}{\query'}%
      =%
      \normY{\query}{\pntA} - \normY{\query}{\pntA}\sin \beta 
      =%
      \normY{\query}{\pntA}
      \pth{\bigl.1 - \sin \beta }.%
    \end{align*}
    In particular, we have
    \begin{math}
        {1 - \sin \beta }%
        \leq%
        {\Bigl.1 - \sin (\pi/2 - \phi/2) }%
        =%
        {\Bigl.1 - \cos ( \phi/2) }%
        \leq%
        1-(1 - (\phi/2)^2/2) %
        = \phi^2 / 8 = \eps/32,
    \end{math}
    since $\cos(x) \geq 1-x^2/2$, for $x \leq \pi/4$. We conclude that
    \begin{math}
        \Delta \leq (\eps/32) \normY{\query}{\pntA},
    \end{math}
    which implies that
    \begin{math}
        (1-\eps/32)\distY{\query}{\pntA} \leq \normY{\query}{\query'}.
    \end{math}
    But then, we have
    \begin{math}
        \distSet{\query}{\PntSetC} \leq%
        \distY{\query}{\pntA} \leq%
        \frac{1}{1-\eps/32} \normY{\query}{\query'}%
        \leq%
        (1+\eps/16) \distSet{\query}{\Star},
    \end{math}
    implying the claim.

    The remaining case is that $\query'$ is in
    $\ballY{\cen}{r} \setminus \Cone$. This implies that
    $\distSet{\query}{\Star} \geq \distSet{\query}{\ballY{\cen}{r}
       \setminus \Cone}$, which is at least
    $r \sin \phi \geq r \sqrt{\eps}/4$, since $\sin x \geq x/2$, for
    $x \leq 1/2$. But this is a contradiction to the assumption that
    $\distSet{\query}{\Star} \leq r \sqrt{\eps}/2$.
\end{proof}

\begin{lemma}
    The query algorithm of \secref{alg:detail} returns a
    $(1+\eps)$-\ANN to $\distSet{\query}{\Star}$, for any
    $\eps \in (0,1)$.
\end{lemma}

\begin{proof}
    If $\distSet{\query}{\Star} \geq r \sqrt{\eps}/2$ then the claim
    follows from \lemref{far}. Otherwise, consider the point
    $\pnt \in \PntSet$, such that $\nnY{\query}{\Star}$ lies on
    $\cen \pnt$. If $\distY{\cen}{\pnt} \leq r$, then by
    \lemref{choose}, the $(1+\eps/4)$-\ANN query on
    $\ballY{\cen}{r} \cap \PntSet \subseteq \PntSet$ would return the
    desired \ANN. (Formally, the result is worse by a factor of
    $(1+\eps/4)$, which is clearly smaller than the desired threshold
    of $(1+\eps)$.) Note, that the algorithm performs \ANN query on
    $\PntSet$, and this resolves this case.

    As such, we remain with the long case, that is
    $\distY{\cen}{\pnt} > r$. But then,
    $\query' = \nnY{\query}{\Star}$ must lie inside $\ballY{\cen}{r}$,
    where $\Star = \starY{\cen}{\PntSet}$. As such, in this case
    $\nnY{\query}{\Star} =
    \nnY{\query}{\starY{\cen}{\PRZ{\cen}{r}{\PntSet}}}$. As such,
    again by \lemref{choose}, the $(1+\eps/4)$-ANN to
    $\PRZ{\cen}{r}{\PntSet}$ is the desired approximation.  Note, that
    the algorithm performs explicitly an \ANN query on this point-set,
    implying the claim.
\end{proof}

\mypar{Running time analysis and result.}  The algorithm built the
data structure of \lemref{prefix} once, and we performed
$O( 1/\eps^2 )$ \ANN queries on it. This results in
$O( (\log n)/\eps^2)$ queries on the original \ANN data structure.  We
thus conclude the following:

\begin{lemma}[\ANN in a star]
    \lemlab{qua:qua}%
    Let $\PntSet = \brc{\pnt_1, \ldots, \pnt_n}$ be a set of $n$
    points in $\Re^d$ and let $0<\eps \leq 1$ be a parameter. Given a
    data structure that can answer $(1+\eps)$-\ANN queries for $n$
    points in $\Re^d$, using $\spY{n}{\eps}$ space and
    $\tmQueryY{n}{\eps}$ query time, then one can build a data
    structure, such that for any query point $\query$, it returns the
    $(1+\eps)$-\ANN to $\query$ in $\Star = \starY{\cen}{\PntSet}$.
    The space needed is $O(\spY{n}{\eps} \log n)$, and the query time
    is $O(\tmQueryY{n}{\eps} \eps^{-2} \log n)$.
\end{lemma}

\subsection{The result -- approximate nearest induced segment}

By building the data structure of \lemref{qua:qua} around each point
of $\PntSet$, and querying each of these data structures, we get the
following result.

\begin{theorem}
    \thmlab{a:n:n:i:segment}%
    Let $\PntSet = \brc{\pnt_1, \ldots, \pnt_n}$ be a set of $n$
    points in $\Re^d$ and let $0<\eps \leq 1$ be a parameter. Given a
    data structure that can answer $(1+\eps)$-\ANN queries for $n$
    points in $\Re^d$, using $\spY{n}{\eps}$ space and
    $\tmQueryY{n}{\eps}$ query time, then one can build a data
    structure, such that for any query point $\query$, it returns a
    segment induced by two points of $\PntSet$, which is
    $(1+\eps)$-\ANN to the closest such segment.  The space needed is
    $O(\spY{n}{\eps} n \log n)$, and the query time is
    $O(\tmQueryY{n}{\eps} n \eps^{-2} \log n)$.
\end{theorem}

%%%%%%%%%%%%%%%%%%%%%%%%%%%%%%%%%%%%%%%%%%%%%%%%%%%%%%%%%%%%%%%%%%%%%%%%%%%%
%%%%%%%%%%%%%%%%%%%%%%%%%%%%%%%%%%%%%%%%%%%%%%%%%%%%%%%%%%%%%%%%%%%%%%%%%%%%
%%%%%%%%%%%%%%%%%%%%%%%%%%%%%%%%%%%%%%%%%%%%%%%%%%%%%%%%%%%%%%%%%%%%%%%%%%%%

\subparagraph*{Acknowledgements.}
We thank Arturs Backurs for useful discussions on how the problem relates to the Hopcroft's problem.

% \hypersetup{urlcolor=black}%
%\BibLatexMode{\printbibliography}

%\BibTexMode{%
   \bibliographystyle{alpha}%
   \bibliography{linear_reg}%

%   \bibliography{shortcuts,geometry}%
%}

\appendix

%%%%%%%%%%%%%%%
\section{Connection to Hopcroft's problem} \seclab{sec:hopcroft}
In the Hopcroft's problem, we are given two sets $U$ and $V$, each consisting of $N$ vectors in $\mathbb{R}^d$ and the goal is to check whether there exists $u\in U$ and $v\in  V$ such that $u$ and $v$ are orthogonal. 

Given an instance of the Affine \SLR with $n$ points and $k=4$, we proceed as follows. Suppose that the input is a set of $n$ points in $\Re^d$ and the query is the origin $\origin$. Moreover, suppose that $d=4$. So the goal is to decide whether there exist four points $a, b, c, d$ such that the three dimensional flat that passes through them also passes through the origin. This is equivalent to checking whether the determinant of the matrix which is formed by concatenating these four points as its columns, is zero or not. We can pre-process pairs of points to solve it fast. 

Take all ${n \choose 2}$ pairs of points $a=(a_1, a_2, a_3, a_4)$ , $b=(b_1, b_2, b_3, b_4)$ and preprocess them by constructing a vector $u$ in 24 dimensional space such that $u=(a_1b_2, -a_1b_2, -a_1b_3, a_1b_3, \cdots , -a_4 b_3 , a_4b_3)$ and let $U$ be the set of such vectors $u$. Also, for each pair of points $c=(c_1, c_2, c_3, c_4)$ , $d=(d_1, d_2, d_3, d_4)$, we construct $v=(c_3d_4, c_4d_3, c_2d_4, c_4d_2, \cdots, c_1d_2, c_2d_1)$  and let $V$ be the set of ${n \choose 2}$ such vectors $v$. It is easy to check that the determinant is zero, if and only if the inner product of $u$ and $v$ is zero. 

Thus, we have two collections $U$ and $V$ of $N = {n\choose 2}$ vectors in $\Re^{24}$ and would like to check whether there exists two points, one from each collection, that are orthogonal.
So any lower bound better than $\omega (n^{k/2}) = \omega (n^{2})$ would imply a super linear lower bound $\omega(N)$ for the Hopcroft's problem.
\end{document}